\documentclass[11pt,a4paper]{article}
\usepackage[utf8]{inputenc}
\usepackage[T1]{fontenc}
\usepackage{mathptmx} 
\usepackage[pdftex]{graphicx} 
\usepackage{geometry}
\usepackage{booktabs}
\usepackage{tabularx}
\usepackage[english]{babel} 
 \usepackage{amsthm}
\usepackage{calc} 
\usepackage{enumitem} 
\usepackage{makecell}
\usepackage{newunicodechar}
\newunicodechar{̀}{\`{}}

\usepackage[]{biblatex}
\usepackage{amsmath}
\usepackage{amsfonts}
\usepackage{amssymb}
\usepackage{xspace}
\usepackage{lmodern}
\usepackage{algorithm}
\usepackage{algpseudocode}
\usepackage{rotating}
\usepackage{tikz}
\usetikzlibrary{arrows.meta,decorations.pathreplacing,calc,positioning}
\usepackage{placeins}
\usetikzlibrary{decorations}
\usepackage{csquotes}
\usepackage{caption}
\usepackage{subcaption}
\usepackage[normalem]{ulem}

\usepackage{graphicx}

\newtheorem{theo}{Theorem}[section]

\newtheorem{rque}{Remark}[section]

\usepackage[dvipsnames]{xcolor} 

\newcommand{\Aspace}{\mathcal{A}}
\newcommand{\Sspace}{\mathcal{S}}

\usepackage[pdftex,linkcolor=black,pdfborder={0 0 0}]{hyperref}
\usepackage{hyperref}
\usepackage{cleveref}
\usepackage{caption}
\setlist[enumerate]{itemsep=0mm}
\usepackage{orcidlink}

\def\1{\mbox{1\hspace{-0.25em}l}}

\usepackage{academicons}
\usepackage{xcolor}
\def\orcid#1{\kern .08em\href{https://orcid.org/#1}{\includegraphics[keepaspectratio,width=0.7em]{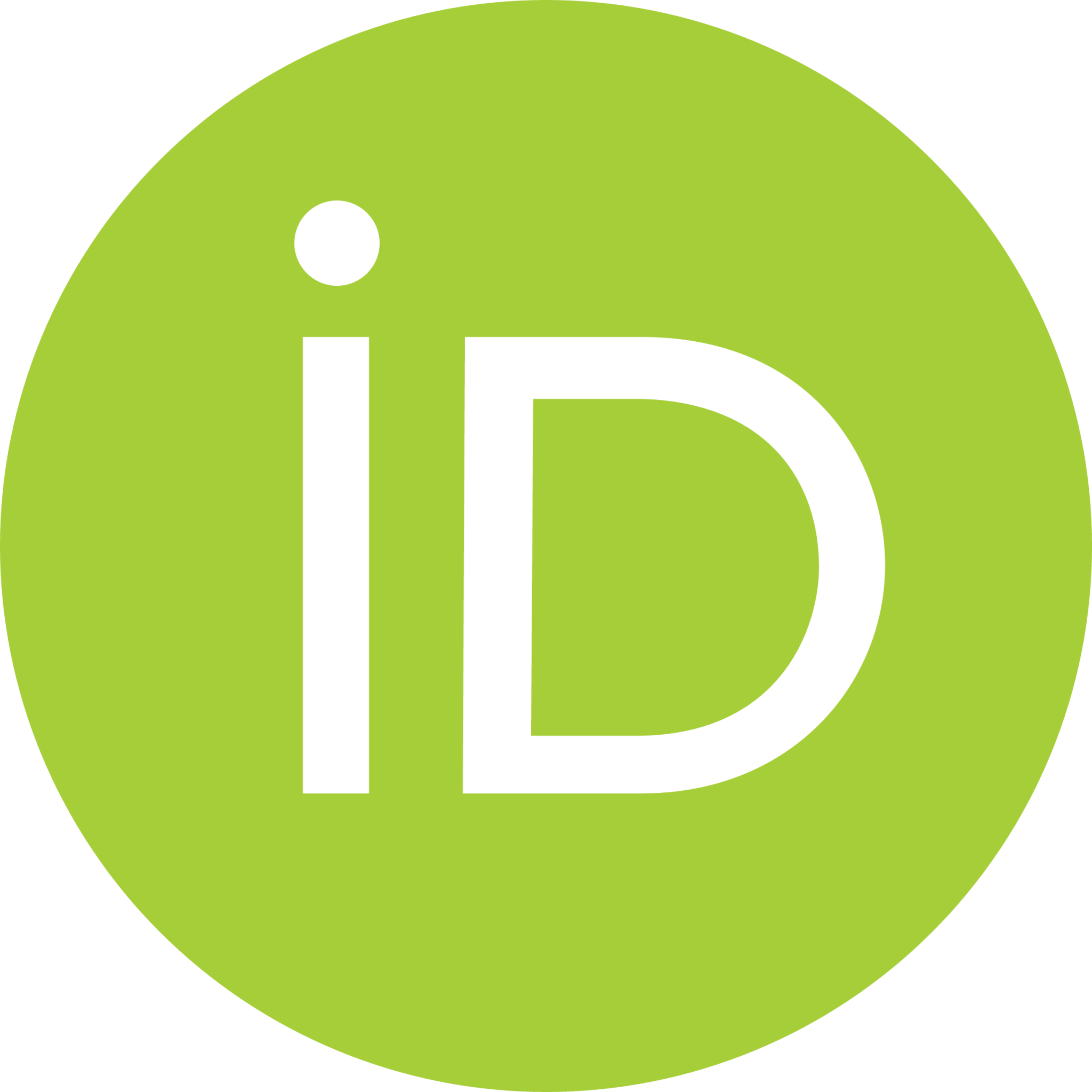}}}
\hypersetup{
  pdfsubject = {OptimalExecutionDDQN},
  pdftitle = {DDQNOptimalExecution},
  pdfauthor = {Andrea Macrì}
}

\begin{document}

\title{Reinforcement Learning in Queue-Reactive Models: \\Application to Optimal Execution }
\author{
Tomas \textsc{Espana}$^{1}$\thanks{tomas.espana@princeton.edu}~~
Yadh \textsc{Hafsi}$^{2}$\thanks{yadh.hafsi@polytechnique.edu, \orcid{https://orcid.org/0009-0001-0686-0349}}~~
Fabrizio 
\textsc{Lillo}$^{3}$\thanks{fabrizio.lillo@sns.it, \orcid{https://orcid.org/0000-0002-4931-4057}}~~
Edoardo \textsc{Vittori}$^{4}$\thanks{edoardo.vittori@intesasanpaolo.com}\vspace{0.3cm}\\
\small $^{1}$ ORFE, Princeton University, Princeton, NJ, USA\\
\small $^{2}$ CMAP, École Polytechnique, Palaiseau, France\\
\small $^{3}$ Scuola Normale Superiore, Pisa, Italy\\
\small $^{4}$ Intesa Sanpaolo, Milan, Italy\\
}

\date{\today}
\maketitle
\begin{abstract}
 We investigate the use of Reinforcement Learning for the optimal execution of meta-orders, where the objective is to execute incrementally large orders while minimizing implementation shortfall and market impact over an extended period of time. Departing from traditional parametric approaches to price dynamics and impact modeling, we adopt a model-free, data-driven framework. Since policy optimization requires counterfactual feedback that historical data cannot provide, we employ the Queue-Reactive Model to generate realistic and tractable limit order book simulations that encompass transient price impact, and nonlinear and dynamic order flow responses. Methodologically, we train a Double Deep Q-Network agent on a state space comprising time, inventory, price, and depth variables, and evaluate its performance against established benchmarks. Numerical simulation results show that the agent learns a policy that is both strategic and tactical, adapting effectively to order book conditions and outperforming standard approaches across multiple training configurations. These findings provide strong evidence that model-free Reinforcement Learning can yield adaptive and robust solutions to the optimal execution problem.
\end{abstract}
\hspace{10pt}
\\
\noindent {\bf Keywords~:} Optimal Execution, Reinforcement Learning, Queue-Reactive Model, Limit Order Book, Market Microstructure, Price impact.
\section{Introduction}
\label{intro}
Executing large orders efficiently is a fundamental challenge in electronic financial markets. Large transactions, typically initiated by institutional investors such as banks, asset managers, hedge funds, or proprietary trading firms, often consume visible liquidity across multiple price levels of the limit order book (LOB), generating significant price impact. This impact is manifested through both immediate price shifts and subsequent order flow reactions and motivates the need for execution strategies that optimally balance market impact against timing risk. Consequently, the design of optimal execution strategies has become a central topic in market microstructure and algorithmic trading research.

Several studies have deepened the theoretical and empirical understanding of market impact. Refs \cite{bouchaud2003fluctuations} and \cite{bouchaud2009markets} documented the transient and concave nature of impact, highlighting its origin in the interplay between order flow, liquidity consumption, and market participant behavior. These empirical findings motivated the introduction of new constraints and mechanisms in the optimal execution problem. The first formal treatment of the problem is due to \cite{optimal_bertsimas_1998}, who derive closed-form solutions under specific assumptions using dynamic programming. This framework was extended by \cite{almgren}, who introduced both permanent and temporary market impact components, as well as a risk-aversion parameter. Subsequent research has generalized the Almgren-Chriss model along several dimensions (see \cite{optimal_almgren_2003, optimal_huberman_2005, optimal_gatheral_2011, optimal_almgren_2012, optimal_obizhaeva_2013, optimal_cheridito_2014, gueant2015general, optimal_cartea_2015, incorporating_cartea_2016, optimal_curato_2017, optimal_digiacinto_2022, optimal_cartea_2022, chevalier2024, chevalier2025}). Notably, \cite{gatheral2012transient} established that transient impact models must satisfy a no-dynamic-arbitrage condition, which tightly links the functional form of impact decay to the underlying price dynamics. This result rules out many ad-hoc impact kernels and provides structural constraints on any admissible transient impact model. \cite{optimal_obizhaeva_2013} further developed this direction by introducing a continuous-time propagator model in which impact decays through limit-order-book resilience.

Although these developments significantly enriched our understanding of execution costs and impact mechanisms, they still rely on strong parametric and structural assumptions that might be in contrast with the true data generating process and might be difficult to calibrate empirically. Model parameters such as impact coefficients, volatility processes, and resilience rates remain highly context-dependent and vary across time and assets, limiting the robustness of these approaches. In fact, traditional methods for solving stochastic control problems face limitations. Closed-form solutions are rare and require restrictive assumptions on the model dynamics and objective function. Numerical approaches typically involve solving the Hamilton-Jacobi-Bellman (HJB) equation, quasi-variational inequalities, or backward stochastic differential equations (BSDEs). These formulations often rely on viscosity solution theory to establish existence, uniqueness, and regularity of the value function. Yet, the associated numerical schemes, such as Monte Carlo methods for BSDEs and finite-difference methods for PDEs, suffer from the curse of dimensionality. As a result, these techniques are generally limited to problems with low-dimensional state spaces.

\paragraph{Related Works.} To address these limitations, reinforcement learning methods \cite{reinforcement_sutton_1998} have emerged as data-driven and model-free alternatives for sequential decision problems. Unlike classical approaches, reinforcement learning (RL) does not require strong parametric assumptions on market dynamics. In the context of optimal execution, RL provides a flexible framework to learn trading policies directly from simulated or historical market data. The first application of RL to optimal execution was introduced by \cite{reinforcement_nevmyvaka_2006}, who trained a tabular Q-learning agent to minimize implementation shortfall using historical trade data. Building on this, \cite{hendricks_2014} proposed a hybrid approach that combined the Almgren–Chriss (AC) model with Q-learning, allowing the agent to adjust its execution trajectory dynamically based on the observed market state. To overcome the limitations of tabular methods in high-dimensional settings, \cite{double_ning_2018} employed Double Deep Q-Networks (DDQN) (see \cite{humanlevel_mnih_2015,vanhasselt_2016}), integrating deep neural networks with Q-learning to generalize across continuous state spaces and mitigate value overestimation bias. More recently, \cite{macri_2024} applied DDQN within the AC framework under time-varying liquidity conditions, demonstrating improved performance relative to benchmark strategies, while \cite{endtoend_lin_2020} employed Proximal Policy Optimization (PPO) (see \cite{proximal_schulman_2017}) to train agents that learn directly from limit order book data using sparse reward structures.

Existing approaches rely on specific market impact assumptions and on historical data, which do not capture how market conditions evolve in response to trading actions. It therefore remains unclear what reinforcement learning algorithms can learn in more realistic, microstructure-based settings that incorporate endogenous market impact. In this work, we address this question using a limit order book (LOB) simulator that models both direct market impact, from liquidity consumption, and indirect impact, from participants’ reactions to trades. Many market simulators have been proposed in the literature, including agent-based models \cite{abides_byrd_2020, minimal_alfi_2009, econophysics_chakraborti_2011, agentbased_hamill_2015} and generative models \cite{ generating_li_2020, learning_coletta_2022, limit_cont_2023, generative_nagy_2023, lobbench_nagy_2025}. Recent studies have explored such simulation-based frameworks \cite{hafsi2024optimal,cheridito2025reinforcement}. While these approaches can capture complex market behaviors, they can be more difficult to implement and calibrate, and, more importantly, they often lack analytical tractability. 

\paragraph{Main Contributions.} We employ the Queue Reactive Model (QRM) (see \cite{huang_2015}). The QRM provides a realistic yet tractable description of short-term limit-order-book dynamics. Instead of assuming that price changes follow a simple diffusion, it treats the best bid and ask as queues of standing orders that evolve through limit-order arrivals, cancellations, and market-order executions. The rates at which these events occur depend on the current state of the book, particularly the size imbalance between the bid and ask queues, which means that the model naturally links order-flow pressure to short-term price movements. When either queue is depleted, the price moves by one tick, and new queue sizes are redrawn from empirical distributions calibrated to market data. This simple mechanism reproduces key microstructure features such as liquidity resilience, mean-reverting price behavior, and transient price impact: a trade that consumes liquidity temporarily shifts the state of the book, altering subsequent order flow and gradually decaying as liquidity replenishes. Because QRM captures both direct and indirect market impact within a statistically grounded framework, it serves as a practical simulator of endogenous market reactions. 

The goal of this paper is to propose a methodology to build a theoretically grounded and microstructure-consistent training and testing environment for RL. Specifically, the QRM is analytically tractable and ergodic, reproducing key stylized facts of LOB dynamics such as resilience, transient impact, and order-flow imbalance. We use model parameters calibrated from real market data and train an RL agent in this environment to learn execution strategies that generalize across endogenous market states. In contrast to existing approaches, our framework captures both direct impact (via liquidity consumption) and indirect impact (via feedback effects on order flow), offering a more realistic and interpretable platform for learning execution policies. 

\vspace{0.3cm}

\begin{figure}[H]
    \centering
    \begin{tikzpicture}[
        node distance=5.5cm,
        every node/.style={font=\sffamily, rounded corners, draw=black, thick, minimum height=1.4cm, minimum width=4cm, align=center},
        arrow/.style={->, thick}
    ]
        \node (step1) {Calibrate QRM \\ on real market data};
        \node (step2) [right of=step1] {Learn optimal execution strategy \\ with RL in ergodic environment};
        \node (step3) [right of=step2] {Deploy in real \\ market environment};

        \draw[arrow] (step1) -- (step2);
        \draw[arrow] (step2) -- (step3);
    \end{tikzpicture}
    \caption{Proposed methodology: learning optimal execution strategies using RL.}
    \label{fig:methodology}
\end{figure}
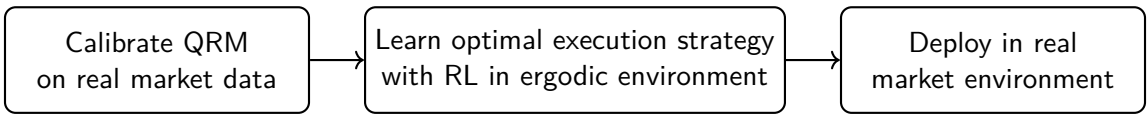

This paper is organized as follows. Section \ref{sec:opt_exec_pb} formalizes the optimal execution problem and establishes the notation used throughout. Section \ref{sec:QRM} introduces the Queue-Reactive Model (QRM), which provides the market environment.  Section \ref{sec:mdp} embeds this problem within a Reinforcement Learning (RL) and Markov Decision Process (MDP) framework, describing how learning-based agents interact with the market. Finally, Section \ref{sec:results} presents and discusses the results of our experiments.

\section{Problem Formulation }\label{sec:opt_exec_pb}
The optimal execution problem consists in executing a trade of $X_0$ shares over a fixed time horizon $[0, T]$. We consider here the problem of optimally executing a buy metaorder. When constructing such a strategy, it is essential to account for the immediate transaction impact of trades, their temporary price effects, and the potential long-term consequences arising from persistent market responses. A terminal penalty term may also be introduced to reflect the cost associated with failing to complete the purchase by the end of the horizon.

In a discrete-time setting with $N{+}1$ decision points, the purchasing strategy is formulated as a sequential decision process, in which the trader determines the quantity $\Delta x_{\tau_k}$ to buy at each time step $\tau_k\!=\!kT/N$, for $k\!\in\!\{0, \ldots, N\}$, with $\tau_0\!=\!0$ and $\tau_N\!=\!T$. The cumulative purchases form a trajectory $\{x_{\tau_0}, \ldots, x_{\tau_N}\}$, where $x_{\tau_k}$ denotes the cumulative number of shares acquired by time $\tau_k$. By construction, $x_0\!=\!0$, and full completion requires $x_T\!=\!X_0$ at the terminal time. Let $P_{\tau_k}$ denote the average execution price for the purchase $\Delta x_{\tau_k}\!=\!x_{\tau_k}\!-\!x_{\tau_{k-1}}$, which depends on both the current and all preceding trades and incorporates all transaction costs. The objective of a risk-neutral trader is to minimize the expected total cost of purchase over the horizon $T$:
\begin{equation}
\label{control_problem_purchase}
\begin{gathered}
\min_{x \in \mathcal{A}} \mathbb{E}\left[\sum_{k=0}^N P_{\tau_k} \Delta x_{\tau_k}\right],\\
\text{where} \quad 
\mathcal{A} = \Big\{(x_{\tau_0}, x_{\tau_1}, \ldots, x_{\tau_N})\in \mathbb{R}^{N+1}_+ : \sum_{k=0}^N \Delta x_{\tau_k} = X_0 \Big\}.
\end{gathered}
\end{equation}
The cost functional in Equation \eqref{control_problem_purchase} is directly related to the Implementation Shortfall (IS) measure introduced by \cite{Perold1988}, which quantifies  the deviation between the realized cost of execution and the cost that would have been incurred had all shares been purchased instantaneously at the initial market price $P_0$. Formally, for a purchasing strategy, the IS is defined as
\begin{equation*}
\text{IS} = \sum_{k=0}^N P_{\tau_k} \, \Delta x_{\tau_k} - X_0 P_0,
\end{equation*}
where the first term represents the realized cost of execution and the second term 
corresponds to the paper cost associated with immediate execution at $P_0$. Minimizing the expected total cost is therefore equivalent to minimizing the expected IS, since the benchmark term $X_0 P_0$ is constant and independent of the trading strategy. 
Economically, this formulation captures the fundamental trade-off between 
\emph{market impact} and \emph{timing risk} as in \cite{almgren}: executing too quickly increases costs through adverse price impact, 
while executing too slowly exposes the trader to unfavorable price movements. Hence, minimizing the expected IS is equivalent to finding the purchasing trajectory 
that optimally balances liquidity consumption and exposure to price risk.

\section{The Queue-Reactive Model}\label{sec:QRM}
The QRM, introduced by \cite{huang_2015}, provides a stochastic representation of the limit order book specifically tailored to large-tick assets.  This mechanism enables the model to reproduce the stylized facts
in high-frequency markets, including the persistence of order-flow imbalance, asymmetric liquidity profiles, and mean-reverting mid-price dynamics. 
\subsection{Description of the Market Simulation}
At each time $t$, the state of the LOB is represented by a $2K$-dimensional vector
$$X(t) = (q_{-K}(t), \ldots, q_{-1}(t), q_1(t), \ldots, q_K(t)),$$
where $K$ denotes the number of visible price levels on each side of the book. The quantity $q_i(t)$ denotes the standing volume at level $Q_i$ at time $t$ priced at
$$p_i = p_{\mathrm{ref}} + \frac{i\,\delta}{2}, \quad \forall i \in \{-K, \ldots, -1, 1, \ldots, K\},$$
with $\delta$ the tick size and $p_{\mathrm{ref}}$ an unobservable \emph{reference price} centered within the LOB. By convention, $Q_{-i}$ denotes a bid level and $Q_i$ an ask level. Formally, $X(t)$ evolves as a continuous-time Markov jump process taking values in $\mathbb{N}^{2K}$ with generator matrix $\mathcal{L}$ specified by 
\begin{align}
    & \mathcal{L}_{l, l+e_i} = f_i(l),\\
    & \mathcal{L}_{l, l-e_i} = g_i(l),\\
    & \mathcal{L}_{l, l} = - \sum_{p \neq l} \mathcal{L}_{l,p},\\
    & \mathcal{L}_{l,p} = 0 \quad \text{otherwise}, 
\end{align}
with $l=(l_{-K}, \ldots, l_{-1}, l_1, \ldots, l_K)\in\mathbb{N}^{2K}$ and $e_i$ denoting the vector $i$ of the canonical basis of $\mathbb{R}^{2K}$. Each queue $q_i(t)$ evolves as a one-dimensional birth–death process governed by the state-dependent intensities $\lambda_i^L(q_i(t))$ for limit order arrivals, $\lambda_i^M(q_i(t))$ for market orders consuming liquidity at level $i$, and $\lambda_i^C(q_i(t))$ for order cancellations. Note there is no bid-ask distinction as we suppose bid-ask symmetry. The transition rates at time $t$ are
$$f_i(X(t)) = \lambda_i^L(q_i(t)),~~\textrm{and}~~ g_i(X(t)) = \lambda_i^M(q_i(t)) + \lambda_i^C(q_i(t)),$$
for all $i\in[-K,\dots,-1,1,\dots,K]$. Conditionally on the current state of the book, order arrivals at each level follow independent Poisson processes. The dependence of intensities on queue sizes induces both auto and cross correlations in the order flow, producing realistic microstructural dynamics.
\begin{rque}
  We retain Model~1 of \cite{huang_2015}, where the intensities $\lambda_i^L$, $\lambda_i^M$, and $\lambda_i^C$ depend solely on the size of the corresponding queue $q_i(t)$. Empirical calibration in \cite{huang_2015} shows that this specification captures most of the variation in order‐flow intensities, with only minor gains in likelihood from adding neighboring queues or imbalance. Later studies \cite{ergodicity_huang_2017} confirm that richer dependencies mainly improve qualitative realism, while the overall quantitative fit remains comparable.
\end{rque}

\subsection{Invariant Distribution and Ergodicity}
 
 Under mild assumptions,~\cite{huang_2015} prove that the $2K$-dimensional Markov jump process $X$ is ergodic. This means that there exists a unique invariant probability measure $\pi$ and that the process converges to it from any start.
 \begin{theo}[Ergodicity]\label{thm:qrm-ergodic}
Assume that
\begin{enumerate}[label=(\roman*)]
\item there exist $C_{\mathrm{bound}}\in\mathbb{N}$ and $\delta_0>0$ such that, for all $i\in[-K,\dots,-1,1,\dots,K]$ and any $p=(p_{-K}, \ldots, p_{-1}, p_1, \ldots, p_K)\in\mathbb{N}^{2K}$ with $p_i>C_{\mathrm{bound}}$, we have 
      $$f_i(p)-g_i(p)\le -\delta_0;$$
\item there exists $H>0$ such that $\sum_{i} f_i(p)\le H$ for every state $p\in\mathbb{N}^{2K}$.
\end{enumerate}
Then $X$ is non-explosive, irreducible, positive recurrent, and therefore
ergodic with a unique invariant probability measure $\pi$.\end{theo}
\begin{proof}
    See Theorem $2.1$ in \cite{huang_2015}.
\end{proof}
The invariant distribution $\pi_i$ of the limit $Q_i$ is given by the intensities of the process introduced previously. Define the arrival/departure ratio $\rho_i$ by 
\begin{align}
    \rho_i(n) = \frac{\lambda_i^L(n)}{\lambda_i^C(n+1) + \lambda_i^M(n+1)}.
\end{align}
The invariant distribution satisfies
\begin{align}
     \pi_i(0) = \left( 1 + \sum_{n=1}^\infty \prod_{j=1}^n \rho_i(j-1) \right)^{-1},~~\textrm{and}~~ \pi_i(n) = \pi_i(0) \prod_{j=1}^n \rho_i(j-1).
\end{align}
\subsection{Price Dynamics}
\label{financial models}
The process $X$ described above characterizes the LOB dynamics for a fixed reference price $p_{\mathrm{ref}}$. To endogenize price movements, two additional parameters are introduced, $\theta$ and $\theta^{\mathrm{reinit}}$. Whenever the mid-price $p_{\mathrm{mid}}$ changes, the reference price $p_{\mathrm{ref}}$ moves by one tick in the same direction with probability $\theta$, provided that the corresponding best queue is depleted ($q_{\pm1}\!=\!0$). Following such a reference-price change, the LOB state is either redrawn from its invariant distribution $\pi$, centered around the new $p_{\mathrm{ref}}$, with probability $\theta^{\mathrm{reinit}}$, or the standing volumes are shifted accordingly with probability $1\!-\!\theta^{\mathrm{reinit}}$. Figure \ref{fig:lob_tree} illustrates these mechanisms when a market order depletes the best ask volume.
\begin{figure}[H]
\centering
\begin{tikzpicture}[>=Latex,thick]
  \def\H{6.5}       
  \def\Y{-0.8}      
  \def\DX{5.1}      

  \path[opacity=0] (-\DX-1.6,\H+1.2) rectangle (\DX+1.6,\Y-1.0);

  \tikzset{edge/.style={-{Latex[length=1.8mm,width=1mm]},
          line width=0.8pt,shorten >=4pt,shorten <=4pt}}

\node (F1) at (0,\H) {%
  \begin{tikzpicture}[x=1.15cm,y=0.5cm,>=Latex,thick,scale=0.66]
    \definecolor{mplC0}{HTML}{1F77B4} 
    \definecolor{mplC1}{HTML}{FF7F0E} 
    \definecolor{mplC3}{HTML}{D62728} 
    \colorlet{axiscol}{black!70}

    \colorlet{bidcol}{mplC0}
    \colorlet{bidline}{mplC0!60!black}

    \colorlet{askcol}{mplC3}
    \colorlet{askline}{mplC3!60!black}

    \colorlet{askdepcol}{mplC1}
    \colorlet{askdepline}{mplC1!60!black}

    \def\barw{0.42}
    \def\bidheights{4.0,5.2,3.2}
    \def\askheights{1.1,2.2,4.6}

    \draw[axiscol,line width=0.9pt] (-3.2,0) -- (4.2,0);
    \draw[axiscol,line width=0.9pt,-{Latex[length=3mm]}] (4.2,0) -- +(0.3,0);

    \draw[axiscol,line width=0.9pt,-{Latex[length=2mm]}] (0,6) -- (0,0);
    \node[fill=white,draw,rounded corners=2pt,inner sep=2pt,font=\footnotesize]
         at (0,6.35) {$p_{\mathrm{ref}}\!=\!p_{\mathrm{mid}}\!=\!1.005$};

    \foreach \h [count=\i] in \bidheights {
      \pgfmathsetmacro\x{-3.5+\i}
      \fill[bidcol,rounded corners=2pt,fill opacity=0.50] (\x-\barw,0) rectangle (\x+\barw,\h);
      \draw[bidline,rounded corners=2pt] (\x-\barw,0) rectangle (\x+\barw,\h);
    }

    \foreach \h [count=\i] in \askheights {
      \pgfmathsetmacro\x{-0.5+\i}
      \ifnum\i=1
        \def\thisfill{askdepcol}
        \def\thisdraw{askdepline}
        \xdef\bestaskx{\x}
        \xdef\bestaskh{\h}
      \else
        \def\thisfill{askcol}
        \def\thisdraw{askline}
      \fi
      \fill[\thisfill,rounded corners=2pt,fill opacity=0.50] (\x-\barw,0) rectangle (\x+\barw,\h);
      \draw[\thisdraw,rounded corners=2pt] (\x-\barw,0) rectangle (\x+\barw,\h);
    }

    \draw[axiscol,line width=0.9pt,-{Latex[length=2mm]}]
      (\bestaskx,{\bestaskh+3}) -- (\bestaskx,\bestaskh);
    \node[fill=white,draw,rounded corners=2pt,inner sep=2pt,font=\scriptsize]
      at (\bestaskx,{\bestaskh+3.5}) {MO at best ask};

    \foreach \x in {-2.5,-1.5,-0.5,0.5,1.5,2.5,3.5}
      \draw[axiscol,line width=0.5pt] (\x,0) -- (\x,-0.35);

    \node[below=8pt,font=\footnotesize] at (-2.5,0) {0.98};
    \node[below=8pt,font=\footnotesize] at (-1.5,0) {0.99};
    \node[below=8pt,font=\footnotesize] at (-0.5,0) {1.00};
    \node[below=8pt,font=\footnotesize] at ( 0.5,0) {1.01};
    \node[below=8pt,font=\footnotesize] at ( 1.5,0) {1.02};
    \node[below=8pt,font=\footnotesize] at ( 2.5,0) {1.03};
    \node[below=8pt,font=\footnotesize] at ( 3.5,0) {1.04};
  \end{tikzpicture}
};

\coordinate (SplitL) at (-1.7,3.7);

\def\leafscale{0.55}

\node[anchor=south] (F3) at (-\DX,\Y) {%
  \begin{tikzpicture}[x=1.15cm,y=0.5cm,>=Latex,thick,scale=\leafscale]
    \definecolor{mplC0}{HTML}{1F77B4} 
    \definecolor{mplC3}{HTML}{D62728} 
    \colorlet{bidcol}{mplC0}
    \colorlet{askcol}{mplC3}
    \colorlet{bidline}{mplC0!60!black}
    \colorlet{askline}{mplC3!60!black}
    \colorlet{axiscol}{black!70}

    \def\barw{0.42}

    \draw[axiscol,line width=0.9pt] (-3.2,0) -- (4.2,0);
    \draw[axiscol,line width=0.9pt,-{Latex[length=3mm]}] (4.2,0) -- +(0.3,0);

    \draw[axiscol,line width=0.9pt,-{Latex[length=2mm]}] (1.0,5) -- (1.0,0);
    \node[fill=white,draw,rounded corners=2pt,inner sep=2pt,font=\footnotesize]
         at (1.0,5.35) {$p_{\mathrm{ref}}\!=\!p_{\mathrm{mid}}\!=\!1.015$};

    \fill[bidcol,rounded corners=2pt,fill opacity=0.50] (-1.5-\barw,0) rectangle (-1.5+\barw,5.2);
    \draw[bidline,rounded corners=2pt]     (-1.5-\barw,0) rectangle (-1.5+\barw,5.2);
    \fill[bidcol,rounded corners=2pt,fill opacity=0.50] (-0.5-\barw,0) rectangle (-0.5+\barw,3.2);
    \draw[bidline,rounded corners=2pt]     (-0.5-\barw,0) rectangle (-0.5+\barw,3.2);
    \fill[bidcol,rounded corners=2pt,fill opacity=0.50] (0.5-\barw,0) rectangle (0.5+\barw,0.8);
    \draw[bidline,rounded corners=2pt]     (0.5-\barw,0) rectangle (0.5+\barw,0.8);

    \fill[askcol,rounded corners=2pt,fill opacity=0.50] (1.5-\barw,0) rectangle (1.5+\barw,1.3);
    \draw[askline,rounded corners=2pt]     (1.5-\barw,0) rectangle (1.5+\barw,1.3);
    \fill[askcol,rounded corners=2pt,fill opacity=0.50] (2.5-\barw,0) rectangle (2.5+\barw,3.5);
    \draw[askline,rounded corners=2pt]     (2.5-\barw,0) rectangle (2.5+\barw,3.5);
    \fill[askcol,rounded corners=2pt,fill opacity=0.50] (3.5-\barw,0) rectangle (3.5+\barw,4.9);
    \draw[askline,rounded corners=2pt]     (3.5-\barw,0) rectangle (3.5+\barw,4.9);

    \foreach \x in {-2.5, -1.5,-0.5,0.5,1.5,2.5,3.5}
      \draw[axiscol,line width=0.5pt] (\x,0) -- (\x,-0.35);

    \node[below=8pt,font=\footnotesize] at (-2.5,0) {0.98};
    \node[below=8pt,font=\footnotesize] at (-1.5,0) {0.99};
    \node[below=8pt,font=\footnotesize] at (-0.5,0) {1.00};
    \node[below=8pt,font=\footnotesize] at ( 0.5,0) {1.01};
    \node[below=8pt,font=\footnotesize] at ( 1.5,0) {1.02};
    \node[below=8pt,font=\footnotesize] at ( 2.5,0) {1.03};
    \node[below=8pt,font=\footnotesize] at ( 3.5,0) {1.04};
  \end{tikzpicture}
};

\node[anchor=south] (F2) at (0,\Y) {%
  \begin{tikzpicture}[x=1.15cm,y=0.5cm,>=Latex,thick,scale=\leafscale]
    \definecolor{mplC0}{HTML}{1F77B4} 
    \definecolor{mplC3}{HTML}{D62728} 
    \colorlet{bidcol}{mplC0}
    \colorlet{askcol}{mplC3}
    \colorlet{bidline}{mplC0!60!black}
    \colorlet{askline}{mplC3!60!black}
    \colorlet{axiscol}{black!70}

    \def\barw{0.42}

    \draw[axiscol,line width=0.9pt] (-3.2,0) -- (4.2,0);
    \draw[axiscol,line width=0.9pt,-{Latex[length=3mm]}] (4.2,0) -- +(0.3,0);

    \draw[axiscol,line width=0.9pt,-{Latex[length=2mm]}] (0.5,6) -- (0.5,0);
    \node[fill=white,draw,rounded corners=2pt,inner sep=2pt,font=\footnotesize]
         at (0.5,6.55) {$p_{\mathrm{mid}}\!=\!1.01$};
    \draw[axiscol,line width=0.9pt,-{Latex[length=2mm]}] (1.0,5) -- (1.0,0);
    \node[fill=white,draw,rounded corners=2pt,inner sep=2pt,font=\footnotesize]
         at (1.0,5.10) {$p_{\mathrm{ref}}\!=\!1.015$};

    \fill[bidcol,rounded corners=2pt,fill opacity=0.50] (-1.5-\barw,0) rectangle (-1.5+\barw,5.2);
    \draw[bidline,rounded corners=2pt]     (-1.5-\barw,0) rectangle (-1.5+\barw,5.2);
    \fill[bidcol,rounded corners=2pt,fill opacity=0.50] (-0.5-\barw,0) rectangle (-0.5+\barw,3.2);
    \draw[bidline,rounded corners=2pt]     (-0.5-\barw,0) rectangle (-0.5+\barw,3.2);

    \fill[askcol,rounded corners=2pt,fill opacity=0.50] (1.5-\barw,0) rectangle (1.5+\barw,2.2);
    \draw[askline,rounded corners=2pt]     (1.5-\barw,0) rectangle (1.5+\barw,2.2);
    \fill[askcol,rounded corners=2pt,fill opacity=0.50] (2.5-\barw,0) rectangle (2.5+\barw,4.6);
    \draw[askline,rounded corners=2pt]     (2.5-\barw,0) rectangle (2.5+\barw,4.6);
    \fill[askcol,rounded corners=2pt,fill opacity=0.50] (3.5-\barw,0) rectangle (3.5+\barw,3.8);
    \draw[askline,rounded corners=2pt]     (3.5-\barw,0) rectangle (3.5+\barw,3.8);

    \foreach \x in {-2.5, -1.5,-0.5,0.5,1.5,2.5,3.5}
      \draw[axiscol,line width=0.5pt] (\x,0) -- (\x,-0.35);

    \node[below=8pt,font=\footnotesize] at (-2.5,0) {0.98};
    \node[below=8pt,font=\footnotesize] at (-1.5,0) {0.99};
    \node[below=8pt,font=\footnotesize] at (-0.5,0) {1.00};
    \node[below=8pt,font=\footnotesize] at ( 0.5,0) {1.01};
    \node[below=8pt,font=\footnotesize] at ( 1.5,0) {1.02};
    \node[below=8pt,font=\footnotesize] at ( 2.5,0) {1.03};
    \node[below=8pt,font=\footnotesize] at ( 3.5,0) {1.04};
  \end{tikzpicture}
};

\node[anchor=south] (F4) at (\DX,\Y) {%
  \begin{tikzpicture}[x=1.15cm,y=0.5cm,>=Latex,thick,scale=\leafscale]
    \definecolor{mplC0}{HTML}{1F77B4} 
    \definecolor{mplC3}{HTML}{D62728} 
    \colorlet{bidcol}{mplC0}
    \colorlet{askcol}{mplC3}
    \colorlet{bidline}{mplC0!60!black}
    \colorlet{askline}{mplC3!60!black}
    \colorlet{axiscol}{black!70}

    \def\barw{0.42}
    \def\bidheights{4.0,5.2,3.2}

    \draw[axiscol,line width=0.9pt] (-3.2,0) -- (4.2,0);
    \draw[axiscol,line width=0.9pt,-{Latex[length=3mm]}] (4.2,0) -- +(0.3,0);

    \draw[axiscol,line width=0.9pt,-{Latex[length=2mm]}] (0,6) -- (0,0);
    \node[fill=white,draw,rounded corners=2pt,inner sep=2pt,font=\footnotesize]
         at (0,6.55) {$p_{\mathrm{ref}}\!=\!1.005$};
    \draw[axiscol,line width=0.9pt,-{Latex[length=2mm]}] (0.5,6) -- (0.5,0);
    \node[fill=white,draw,rounded corners=2pt,inner sep=2pt,font=\footnotesize]
         at (0.5,5.10) {$p_{\mathrm{mid}}\!=\!1.01$};

    \foreach \h [count=\i] in \bidheights {
      \pgfmathsetmacro\x{-3.5+\i}
      \fill[bidcol,rounded corners=2pt,fill opacity=0.50] (\x-\barw,0) rectangle (\x+\barw,\h);
      \draw[bidline,rounded corners=2pt] (\x-\barw,0) rectangle (\x+\barw,\h);
    }
    \fill[askcol,rounded corners=2pt,fill opacity=0.50] (1.5-\barw,0) rectangle (1.5+\barw,2.2);
    \draw[askline,rounded corners=2pt]     (1.5-\barw,0) rectangle (1.5+\barw,2.2);
    \fill[askcol,rounded corners=2pt,fill opacity=0.50] (2.5-\barw,0) rectangle (2.5+\barw,4.6);
    \draw[askline,rounded corners=2pt]     (2.5-\barw,0) rectangle (2.5+\barw,4.6);

    \foreach \x in {-2.5,-1.5,-0.5,0.5,1.5,2.5, 3.5}
      \draw[axiscol,line width=0.5pt] (\x,0) -- (\x,-0.35);

    \node[below=8pt,font=\footnotesize] at (-2.5,0) {0.98};
    \node[below=8pt,font=\footnotesize] at (-1.5,0) {0.99};
    \node[below=8pt,font=\footnotesize] at (-0.5,0) {1.00};
    \node[below=8pt,font=\footnotesize] at ( 0.5,0) {1.01};
    \node[below=8pt,font=\footnotesize] at ( 1.5,0) {1.02};
    \node[below=8pt,font=\footnotesize] at ( 2.5,0) {1.03};
    \node[below=8pt,font=\footnotesize] at ( 3.5,0) {1.04};
  \end{tikzpicture}
};

\draw[edge] (F1.south) -- node[pos=0.55,above left] {$\theta$} (SplitL);
\draw[edge] (F1.south) -- node[pos=0.45,above right] {$1-\theta$} (F4.north);
\draw[edge] (SplitL) -- node[pos=0.45,above left] {$\theta^{\mathrm{reinit}}$} (F3.north);
\draw[edge] (SplitL) -- node[pos=0.55,above right] {$1-\theta^{\mathrm{reinit}}$} (F2.north);

\end{tikzpicture}
\vspace{-0.76cm}
\caption{QRM response to consuming the best ask $q_1$ at time $t$ with tick size $\delta=0.01$. The root node shows a typical LOB state just before the trade at time $t^-$, with volumes drawn from the invariant distribution, while the leaf nodes depict the possible post-trade configurations generated by the QRM dynamics at time $t$.}
\label{fig:lob_tree}
\end{figure}
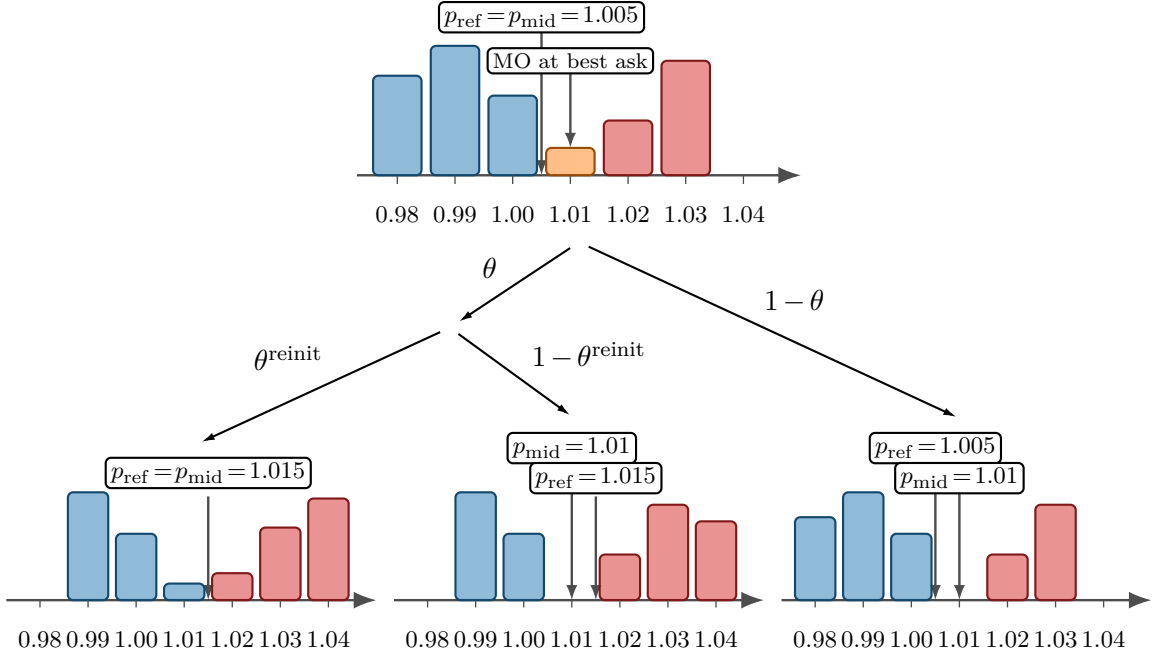
In \cite{huang_2015}, the authors of  interpret parameter $\theta^{\mathrm{reinit}}$ as quantifying the proportion of price changes associated with exogenous information shocks. Following such shocks, market participants typically rebalance their order flows around the new reference price almost instantaneously. The empirical calibration of \cite{huang_2015} using France Télécom (Euronext Paris) data from January $2010$ to March $2012$, yields $\theta\!=\!0.7$ and $\theta^{\mathrm{reinit}}\!=\!0.85$,
 implying that most price adjustments are accompanied by rapid order-book reconfigurations. Although this value may seem high, since volatility models generally attribute around $80\%$ of price variance to endogenous, self-referential mechanisms (see for example \cite{trades_bouchaud_2018}), it can alternatively be understood as the probability that the order book refills following a price movement driven either by market makers or by algorithmic liquidity provision. Henceforth, unless otherwise stated, we adopt the same calibration as \cite{huang_2015} and use the same order-flow intensities.
\subsection{Market impact}\label{sec:mi_refill}
In this Section, we explore the joint role of $(\theta, \theta^{\mathrm{reinit}})$ in shaping liquidity dynamics. To this end, we analyze how the QRM responds when an external trader consumes the entire best ask volume $q_1$, as shown in Figure \ref{fig:lob_tree}. This controlled perturbation shows that different parameter combinations give rise to distinct market-impact regimes, where some configurations lead to prices mean-reverting after the trade, corresponding to a transient impact, while others result in prices continuing to rise on average, corresponding to a permanent impact.

 Let $t_{k}^-$ denote the time just before a buy MO at the best ask that consumes all the best ask volume and $p_{k}$ the associated mid-price after. The index $k$ runs over all the events, and not only the MOs. Furthermore, we assume $p_{\mathrm{ref}}(t_{k}^-)\!=\!p_{\mathrm{mid}}(t_{k}^-)$, with the surrounding volumes $q_{\pm i}$ sampled from the invariant distribution, consistent with the ergodicity of the process $X$ under mild conditions. For simplicity, all queues $q_{\pm i}$ are considered non-empty, as this represents the typical configuration of the invariant distribution $\pi$. When the LOB is redrawn from $\pi$, we assume without loss of generality that all queues are non-empty, as the contribution of empty queues cancels out on average by bid–ask symmetry. The expected mid-price jump after consuming the best ask is 
\begin{align}
    \mathbb{E}[\Delta p_k]
    := \mathbb{E}[p_k - p_{k-1}]
    = (1 + \theta\,\theta^{\mathrm{reinit}})\frac{\delta}{2},
    \label{eq:mid_price_jump}
\end{align}
where 
$\delta$ is the tick size. 
 As expected, $\mathbb{E}[\Delta p_k]$ increases with both $\theta$ and $\theta^{\mathrm{reinit}}$, consistent with the behavior observed in Figure \ref{fig:theta_reinit}. 
\begin{figure}[H]
    \centering
    \includegraphics[width=0.45\textwidth]{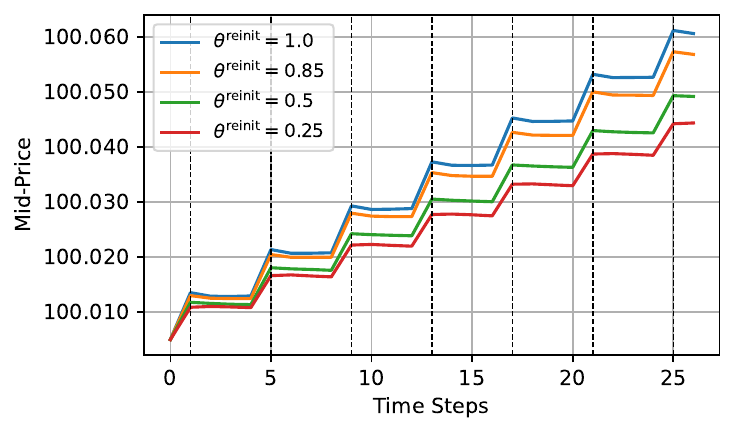}
    \caption{Average mid-price across $20{,}000$ simulations in which a trader systematically buys the entire best ask at fixed time intervals (vertical dashed lines). We set $\theta\!=\!0.7$.}
    \label{fig:theta_reinit}
\end{figure}
Consider now the next event $k\!+\!1$ that updates the LOB. 
With probability $\theta \theta^{\mathrm{reinit}}$, the book is redrawn from its invariant distribution and the mid-price remains unchanged on average. 
With probability $\theta(1\!-\!\theta^{\mathrm{reinit}})$, a bid refill occurs, increasing the price by $\delta/2$; 
conversely, with probability $(1\!-\!\theta)$, an ask refill occurs, decreasing the price by $\delta/2$ and inducing mean reversion. These last two events provide a first-order approximation, as in the large-tick regime the next events are typically bid or ask refills. Averaging over these outcomes yields
\begin{equation}
    \mathbb{E}[\Delta p_{k+1}]
    = \big[\theta(2 - \theta^{\mathrm{reinit}}) - 1\big]\frac{\delta}{2},
    \label{eq:phase_transition}
\end{equation}
which holds on very short time scales only. The sign of $\theta(2 - \theta^{\mathrm{reinit}})\!-\!1$ delineates the post-trade regime: 
a positive value implies that the quote adjustment after the trade is on average in the same direction of the trade, while a negative value indicates a mean reversion of the midprice after the trade.

In order to numerically test this expression, we initialize the LOB by sampling it from the invariant distribution and then we "send" a buy market order that completely depletes the best ask. Here and in the following we set the tick size at $\delta=0.01$. Figure \ref{fig:heatmap_short_term} illustrates the contour plot of the estimated $\mathbb{E}[\Delta p_{k+1}]$ for different values of $\theta, \theta^{\mathrm{reinit}}\!\in\![0.5, 1.0]$. The dashed black line in the figure  denotes the theoretical frontier $\mathbb{E}[\Delta p_{k+1}]\!=\!0$ described by Equation \eqref{eq:phase_transition}, aligning closely with the empirical phase boundary (white region).
\begin{figure}[H]
  \centering
  \begin{minipage}[t]{0.4\textwidth}
    \centering
    \includegraphics[width=\linewidth]{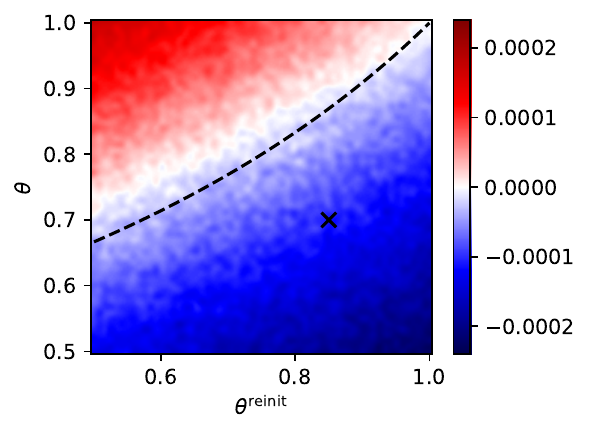}
    \subcaption{Short term.}
    \label{fig:heatmap_short_term}
  \end{minipage}
  \hspace{1.5em}
  \begin{minipage}[t]{0.4\textwidth}
    \centering
    \includegraphics[width=\linewidth]{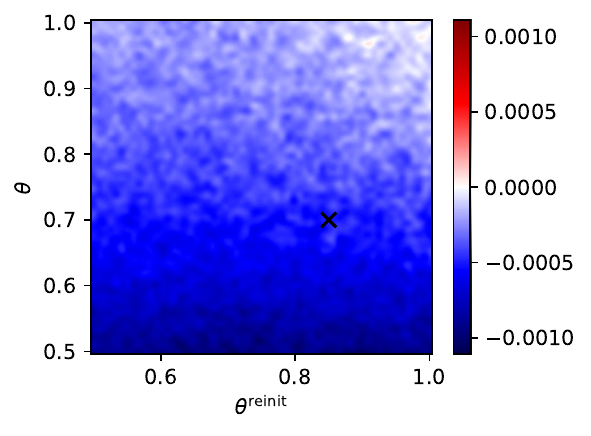}
    \subcaption{Long term.}
    \label{fig:heatmap_long_term}
  \end{minipage}

  \caption{Heatmaps of expected immediate price change $\mathbb{E}[\Delta p_{k+1}]$ (left) and cumulative impact $\mathbb{E}[p_{k+75}\!-\!p_{k}]$ averaged over $20{,}000$ simulations (right) starting from a stationary LOB state with an exogenous trader consuming the best ask. The black cross marks $(\theta, \theta^{\text{reinit}})\!=\!(0.7, 0.85)$.}
  \label{fig:heatmap}
\end{figure}
 
 We now numerically study the long term behavior of the price after an ask depleting market order. The difficulty of course is to establish, for the different parameter configurations, when the asymptotic value of the price has been reached. We first consider a fixed time interval and estimate $\mathbb{E}[p_{k+75}\!-\!p_{k}]$  as a function of $(\theta, \theta^{\text{reinit}})$ (see Figure \ref{fig:heatmap_long_term}). We observe that prices consistently exhibit long-term mean reversion across all parameter configurations. The only exception is the white region in the top-right corner of Figure \ref{fig:heatmap_long_term}, which occurs because for large values of both $\theta$ and $\theta^{\mathrm{reinit}}$ the
 volumes are almost always redrawn from the invariant distribution around the updated reference price. As $\theta$ and $\theta^{\mathrm{reinit}}$ decrease, prices exhibit a stronger mean reversion (blue regions), with the long-term intensity of this effect governed primarily by $\theta$. As shown in Figure \ref{fig:lob_tree}, there are two scenarios after buying the best ask: either the LOB is redrawn from its invariant distribution, or there is a bid-ask refill. When the LOB is redrawn from the invariant distribution, the average mid-price remains constant due to bid–ask symmetry. Thus, this scenario does not contribute to the observed mid-price changes, which are entirely accounted for by the bid-ask refill scenarios (see Appendix \ref{app:refills}).
 
\begin{figure}[H]
  \centering

  \begin{minipage}[t]{0.38\textwidth}
    \centering
    \includegraphics[width=\linewidth]{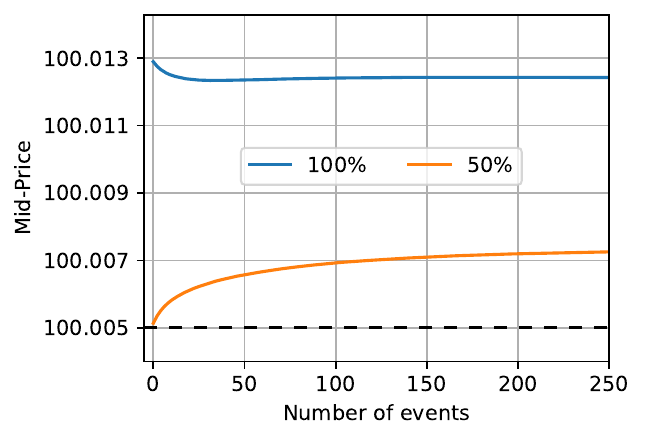}
    \subcaption{$(\theta, \theta^{\mathrm{reinit}})\!=\!(0.7, 0.85)$.}
    \label{fig:combined_market_impact_a}
  \end{minipage}
  \hspace{2em}
  \begin{minipage}[t]{0.38\textwidth}
    \centering
    \includegraphics[width=\linewidth]{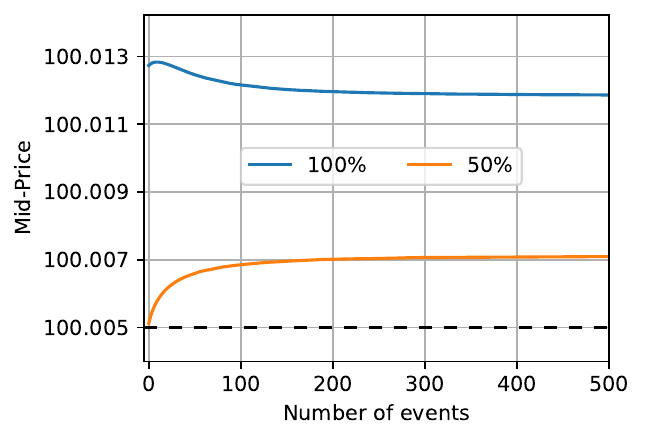}
    \subcaption{$(\theta, \theta^{\mathrm{reinit}})\!=\!(0.9, 0.6)$.}
    \label{fig:combined_market_impact_b}
  \end{minipage}

  \caption{Evolution of $\mathbb{E}[p_k]$ after buying the best ask at $k\!=\!0$, averaged over $1{,}000{,}000$ simulations. The mid-price before the trade is $100.005$ (horizontal dashed line).}
  \label{fig:combined_market_impact}
\end{figure}

To investigate in more detail the price reversion after a trade, we now fix the value of $(\theta, \theta^{\mathrm{reinit}})$ and study the average price dynamics via numerical simulations. As before, we draw the initial state of the LOB from the invariant distribution and then we consider two scenarios: in the first, as above, we send a buy market order of size equal to the best ask, while in the second the size is half of the best ask volume\footnote{Since QRM orders have integer size, we set the market order size to the floor of half the ask volume.}. We make this choice because the RL algorithm used below will have these options (plus the choice of not trading) in the action space. Figure \ref{fig:combined_market_impact} shows the average price trajectories for two configurations of $(\theta, \theta^{\text{reinit}})$. We observe that in both cases, when the market order depletes the best ask, the price displays a clear mean reversion, which can be very slow as in the right top panel. On the contrary, when the size of the market order is half the volume at the best ask (clearly without mechanically moving the price), the price trends in the same direction of the trade. This difference indicates that the impact model associated with the QRM is inherently nonlinear in the trade size, differently from reduced form standard  models used in optimal execution, such as Almgren \& Chriss or the Transient Impact Model\footnote{Interestingly, extensions of the Transient Impact Model with more propagators, such as the one in \cite{Taranto03062018}, are able to reproduce the behavior in Figure \ref{fig:combined_market_impact}. However the closed form solution for the optimal execution problem is not known in this case.} \cite{bouchaud2003fluctuations}. This implies that the optimal trading for the QRM should tactically place market orders of a size which depends, among other things, on the state of the LOB and this makes the problem inherently complex justifying the adoption of deep RL to solve it. Moreover, we expect that the price trajectory after a trade depends also on the absolute volume at the ask and not only on the fraction of it taken by the trade\footnote{We cannot test it directly with this type of simulations because it depends on the correlation between consecutive queue sizes, which is zero when sampling from the invariant distribution. In fact, when the best ask has a small volume, it is likely that the second best ask volume is also small, inducing the possibility of a trending price in the direction of the depleting market order.}. We will add this variable in the state space of the RL algorithm and we will show below that it brings a significant contribution to its performance.
 
\begin{rque}
    From an implementation perspective, the simulator is designed such that when the trader consumes the full best ask, triggering a mid-price change, the QRM reacts identically to a mid-price change generated endogenously within the model, since it does not differentiate between endogenous and exogenous price movements.
\end{rque}

\section{Optimal Execution Setting}\label{sec:mdp}

\subsection{Markov Decision Process Embedding}\label{mdp2}
The execution problem is formulated as a Markov Decision Process (MDP). The MDP is defined as a tuple $\langle\Sspace,\Aspace, \mathbb{P}, \mathcal{R}, \gamma, \mu\rangle$ (see \cite{puterman2014markov}), where $\Sspace$ is the state space, $\Aspace$ the action space, $\mathbb{P}(\cdot|s,a)$ is a Markovian transition model that assigns to each state-action pair $(s,a)$ the probability of reaching the next state $s'$, $\mathcal{R}(s,a)$ is a bounded reward function, $\gamma\!\in\![0,1[$ is the discount factor, and $\mu$ is the distribution of the initial state. 
\paragraph{State space} Each state $s_{\tau_k}$ at time $\tau_k$ for $k\!\in\!\{0,\dots n\}$ is defined as
\begin{align*}
s_{\tau_k}
= \big( 
\, \tau_k,\;
\text{inventory}_{ \tau_k},\;
\text{best ask price}_{\; \tau_k},\;
\text{best bid volume}_{\; \tau_k},\;
\text{best ask volume}_{\; \tau_k} 
\, \big).
\end{align*}
 This constitutes a minimal state representation. The agent is informed of the remaining time, its current inventory, and relevant local market conditions in terms of price and liquidity. In particular, the inclusion of both best bid and best ask volumes allows the agent to infer volume imbalance, a well-established short-term predictor of price movements in market microstructure (see \cite{pulido2024understandingworstkeptsecrethighfrequency}). A state is considered terminal either when the time horizon $T$ is reached or when the agent has fully executed its inventory prior to $T$.

\paragraph{Action space} Since the QRM is designed to model large-tick stocks, the action space is constructed so that the agent may consume at most the best ask volume, as trading beyond the first depth level is prohibitively costly. 
Empirical evidence supports that traders rarely consume more than the first level (see \cite{Pomponio2010}). The action space therefore consists of percentages, where taking an action of, say, $50\%$ at time $\tau_k$ corresponds to buying half of the best available ask at time $\tau_k^-$. 
This formulation enables the RL agent to adjust its trading volume proportionally to prevailing market conditions: executing $50\%$ of the available best when liquidity is high is more profitable than when liquidity is low, while the market impact in both cases remains comparable. 

In the QRM, the state variables $q_i$ at each depth $i$ are expressed in units normalized by the \emph{Average Event Size} $(\text{AES}_i)$, which represents the mean event size across limit, market, and cancel orders observed at level $Q_i$. 
Accordingly, the executable quantities $\Delta x_{\tau_k}$ are integer-valued, $x_{\tau_k}\in \mathbb{N}$. 
In what follows, we focus on a simplified action space,
\begin{align}
\mathcal{A} = \{0\%, 50\%, 100\%\},
\end{align}
which allows the agent either to remain inactive, consume half the entire best ask or the entire best ask, thereby accelerating convergence during training. 

 \paragraph{Reward function} In this work, our goal is to minimize the expected implementation shortfall under the risk-neutral probability measure. The instantaneous reward at time $\tau_k$ is defined as
\begin{align}
r_{ \tau_k} = \Delta x_{\tau_k} (P_0 - P_{ \tau_k})-\alpha \1_{\{\tau_k=T\}}(X_0 - x_T),
\end{align}
where $P_0$ denotes the initial midprice, $\Delta x_{\tau_k}$ the number of shares executed, $\alpha>0$ a positive constant and $P_{\tau_k}$ the execution price, which in our setting is always the best ask. A terminal penalty is applied if the agent fails to fully execute the inventory by the end of the horizon. 
This penalty is equal to the number of shares remaining to be executed, scaled by a final penalty parameter $\alpha$, thereby encouraging complete execution of the $X_0$ shares  before time $T$. For higher risk aversion, the execution horizon $T$ may be shortened to limit price-risk exposure.

\subsection{Reinforcement Learning}\label{sec:rl}
In this paper, we employ Deep Reinforcement Learning to approximate the optimal execution schedule within the QRM model. Specifically, we adopt the Double Deep Q-Network (DDQN) algorithm \cite{vanhasselt_2016}, a model-free, online, off-policy reinforcement learning approach, which provides a model-agnostic solution to the optimal execution problem. 

In this work, we focus on the Double Deep Q-Network (DDQN) algorithm as it provided the best learning results. Algorithm \ref{alg:DDQN_qrm} presents the training procedure for a Double Deep Q-Network (DDQN) agent.
\begin{algorithm}
\caption{DDQN Algorithm in the QRM environment}
\begin{algorithmic}[1]
\Require 
Initialize $Q_{\text{main}}$ (random weights), $Q_{\text{tgt}} \leftarrow Q_{\text{main}}$;\\
Replay memory (size $L$), $\epsilon=1$, batch size $b$, episodes $M$, $c<1$;\\
Set QRM parameters and market state $S_0$, inventory $q_0$.
\For{$i=1$ to $M$} \label{alg:DDQN_outer_loop}
    \State Initialize the QRM simulation \label{alg:DDQN_reset}
    \For{$t=1$ to $N$} \label{alg:DDQN_timestep_loop}
        \State $s_t \gets QRM(t)$ \label{alg:DDQN_state}
        \State $a_t \gets 
        \begin{cases}
            \text{random action} & \text{w.p. } \epsilon\\
            \arg\max\limits_{a} Q_{\text{main}}(s_t,a|\theta_{\text{main}}) & \text{w.p. } 1-\epsilon
        \end{cases}$ \label{alg:DDQN_eps_greedy}
        \State Execute $a_t$ in QRM $\Rightarrow$ new state $s_t$, observe reward $r_t$ \label{alg:DDQN_reward}
        \State Store $(s_t, r_t, a_t, s_{t+1})$ in memory \label{alg:DDQN_store}
        \If{$|\text{Memory}| \ge b$} \label{alg:DDQN_train_cond}
            \State Sample $(s^j_t,r^j_t,a^j_t,s^j_{t+1})_{j=1}^b$ \label{alg:DDQN_sample}
            \State $a^{*,j} = \arg\max\limits_a Q_{\text{main}}(s^j_{t+1},a|\theta_{\text{main}})$ \label{alg:DDQN_vstar}
            \State $y^j = r^j_t + \gamma Q_{\text{tgt}}(s^j_{t+1}, a^{*,j}|\theta_{\text{tgt}})$ \label{alg:DDQN_target}
            \State Update $\theta_{\text{main}}$ minimizing 
            $\mathcal{L} = \frac{1}{b}\sum_j (y^j - Q_{\text{main}}(s^j_t,a^j_t|\theta_{\text{main}}))^2$ \label{alg:DDQN_loss}
        \EndIf
        \If{$t \bmod m = 0$} \label{alg:DDQN_sync_cond}
            \State $\epsilon \leftarrow \epsilon - c$, \quad $\theta_{\text{tgt}} \leftarrow \theta_{\text{main}}$ \label{alg:DDQN_sync}
        \EndIf
    \EndFor
\EndFor
\end{algorithmic}
\label{alg:DDQN_qrm}
\end{algorithm}
Two neural networks are initialized: the main network $Q_{\text{main}}$ for action selection, and the target network $Q_{\text{tgt}}$ for value estimation. At each time step, the agent observes the current state and selects an action via an $\epsilon$-greedy policy that balances exploration and exploitation. The resulting transition and reward are stored in a replay buffer. Once the buffer reaches a sufficient size, mini-batches are sampled to update $Q_{\text{main}}$ by minimizing the Bellman loss using target values computed from $Q_{\text{tgt}}$. Every $m$ steps, the target network is synchronized with the main network and the exploration rate $\epsilon$ is decayed. This iterative process allows the agent to approximate the optimal action–value function while controlling value overestimation. After training, $Q_{\text{main}}$ serves as the agent’s policy for decision-making in the environment. 

We consider a finite-horizon setting with exponentially discounted future rewards by the factor $\gamma$. More specifically, the per-step reward entering the Bellman recursion is
$$r_k := r(s_{\tau_k},a_{\tau_k}) = a_{\tau_k}\,(P_0 - P_{\tau_k}) 
- \alpha\,\mathbf{1}_{\{\tau_k = T\}}\,(X_0 - x_T).$$
A trajectory is a sequence of states, actions, and rewards up to a stopping time $\tau$, i.e.,
\begin{equation*}
(s_{0}, a_{0}, r_{1}, s_{1}, a_{1}, r_{2}, ..., s_{\tau-1}, a_{\tau-1}, r_{\tau}).
\end{equation*}
Given a policy $\mathcal{\pi}$ we can define the \textit{State-Action Value Function} 
\begin{equation}
Q_\mathcal{\pi}(s, a) = \mathbb{E}_{\pi}[\sum_{i=1}^{\tau} \gamma^{i-1} r_{i} | s_0 = s, a_0 = a],
\end{equation}
which represents the expected return from state $s$ if we take action $a$ and then we follow the policy $\pi$ and can be recursively defined by the following Bellman equation \cite{bellman1966dynamic}, 
\begin{equation}
        Q_\pi(s,a) = r(s,a) + \gamma \mathbb{E}_{\substack{s'\sim \mathcal{P}(\cdot|s,a)\\a'\sim\pi(\cdot|s')}}\big[Q_\pi(s',a')\big].
        \label{eq:bellman}
\end{equation}
Solving the MDP means finding the \textit{optimal} policy $\pi^*$ which is the policy that maximizes the objective
\begin{equation*}
    J_\pi := \underset{s_0\sim \mu}{\mathbb{E}_{\pi}}\Big[\sum_{i=1}^{\tau} \gamma^{i-1} r_{i}\Big].
\end{equation*}

\section{Experiments and results}\label{sec:results}

In this section we present the results of our numerical investigations. The aim of the experiments is to study whether the DDQN agent is able to find robust optimal execution strategies without any form of knowledge of the underlying impact model. We focus here on the best-performing configuration, while additional experiments exploring alternative state and action spaces are reported in Appendix \ref{app:reduced_state_spaces} for comparison.

\subsection{Benchmark Strategies} \label{sec:benchmark_strats}
To evaluate the DQN policy, we introduce a set of benchmark execution strategies designed to minimize Implementation Shortfall (IS) and provide a basis for performance comparison.
\paragraph{Baseline Model}
We consider the Time-Weighted Average Price (TWAP) benchmark, in which the trader executes $X_0$ number of shares uniformly over a fixed horizon $[0, T]$ divided into N discrete intervals. The trading rate is constant, given by
$\Delta x^* = \left( \frac{X_0}{N}, \dots, \frac{X_0}{N} \right)$,
or equivalently, the number of executed shares after time step $\tau_k$ is equal to $x^*_{\tau_k} = k\frac{X_0}{N}$, for $k = 0, \dots, N$.
For a risk-neutral trader, this policy coincides with the Almgren–Chriss (A\&C) solution \cite{almgren}, whose objective is to minimize the expected Implementation Shortfall (IS). Under the standard A\&C assumptions of linear permanent and temporary market impact, and asset price dynamics following a Brownian motion with constant volatility, the TWAP is the optimal strategy. 

 \paragraph{The Percentage of Posted Volume Benchmark} We introduce a family of new benchmarks we call Percentage Of Posted Volume (POPV). More precisely, we define $\text{POPV}_i$ as the strategy that purchases a fixed fraction (50\% or 100\%) of the available volume at the best ask every $i$ time steps, while remaining inactive otherwise. The intuition behind this strategy is to exploit pauses in execution during which prices tend to mean revert (see Figure \ref{fig:theta_reinit}). 
\subsection{Training Configuration}
\paragraph{Algorithm parametrization} In our implementation, we employ fully connected feed-forward neural networks comprising 5 layers, each with $30$ hidden units and leaky-ReLU activation functions. We use the ADAM optimizer for optimization. The RL agent is trained on approximately $500,000$ episodes. The epsilon-greedy exploration policy starts at $1.0$ and decays linearly to $0.01$ over the first $3\%$ of training. We set the final penalty $\alpha\!=\!1.0$. The parameters used to calibrate the algorithm are reported in Table \ref{tab:pmts}. The parameters not shown in the table change depending on the experiments and are reported below accordingly.

\begin{table}[H]
\centering
\begin{tabularx}{0.75\textwidth}{@{}l c l c@{}}
\toprule
\textbf{DDQN Parameters} & & \textbf{Model Parameters} & \\ 
\midrule
NN layers                & 5           & Time horizon (T)     & 600~s \\
Hidden nodes             & 30          & Time intervals ($N$) & 25 \\
ADAM lr                 & 1e-4         & Shares to execute ($X_0$)  & 25 \\
Batch size ($b$)              & 1,024          & Final Penalty ($\alpha$) & 1.0 \\
Replay memory ($L$)           & 1e6      & $\theta$ & 0.7 \\
Target update ($m$)         & 1e3      &  $\theta^{\mathrm{reinit}}$ & 0.85 \\
Training episodes ($M$)       & 5e5      &  &  \\
Test episodes ($B$)           & 2e4       & & \\
Discount factor ($\gamma$)     & 0.995    &  & \\
\bottomrule
\end{tabularx}
\vspace{0.3ex}
\begin{minipage}{0.75\textwidth}
\footnotesize \emph{Note.} The target update $m$ is in number of environment steps (not episodes).
\end{minipage}
\caption[b]{Fixed parameters used in the DDQN algorithm.}
\label{tab:pmts}
\end{table}

\paragraph{Feature Normalization}  As the learning model relies on neural networks, all input features are normalized. Time and inventory are linearly rescaled to the interval $[-1, 1]$, while prices and volumes are standardized using z-score normalization. 
 
\subsection{Simulation under Market-Calibrated Dynamics} 
\label{sec:results_2}

\paragraph*{Environment Setup.} In this section, we adopt the calibration proposed in \cite{huang_2015}, with parameters $\theta\!=\!0.7$ and $\theta^{\text{reinit}}\!=\!0.85$, and use identical order-flow intensities. These parameters were calibrated on France Télécom (Euronext Paris) using data from January $2010$ to March $2012$, where the average bid–ask spread was approximately $1.43$ ticks. The agent wants to buy $25$ shares and may remain inactive or purchase half or the entire best ask volume at each decision point, i.e., $\mathcal{A}\!=\!\{0\%, 50\%, 100\%\}$. We simulate $600$ seconds of the QRM with a trader time step of $25$ seconds, meaning the agent can act at $\tau_0\!=\!0$ and subsequently every $25$ seconds. Each of these $25$ second interval is called {\it Trader Step}. This combination of parameters is designed to balance execution urgency with tactical flexibility. The number of shares to execute is large enough to require multiple market interventions, yet not so large relative to the time horizon that the optimal policy collapses into systematic buying at every step. This setup allows us to assess whether the RL agent can learn to wait for favorable market conditions and adapt its trading behavior accordingly. To compare the description in events given above with the one in seconds used here, it is useful to remark that there are on average $7$ events per second .

We benchmark the RL agent against the strategies introduced in Section \ref{sec:benchmark_strats}: TWAP, POPV1, POPV2, POPV3 and POPV4. For a fair comparison, POPV1 and POPV2 take action $50\%$ while POPV3 and POPV4 take action $100\%$. Executing 50\% of the posted volume results in a higher Implementation Shortfall (IS) due to reduced market impact compared to the more aggressive 100\% setting. Moreover, it is necessary to ensure that the entire inventory is executed within the specified time horizon. This constraint justifies using only POPV1 and POPV2 for the $50\%$ action, as the slower execution of POPV3 and POPV4 would prevent full completion. Conversely, in the $100\%$ action, POPV1 and POPV2 become overly aggressive and yield worse IS, motivating the focus on POPV3 and POPV4. To assess statistical significance, we performed a one-sided Welch’s t-test to determine whether the best-performing strategy has a significantly higher average IS than the second-best. Statistical significance is indicated by asterisks: (*) for $p < 0.05$, (**) for $p < 0.01$, and (***) for $p < 0.001$. To ensure comparability across methods, when the agent fails to fully execute its inventory, a final trade is executed at an additional time step and its cost is included in the IS. 

\paragraph*{Learning Dynamics and Q-Value Analysis.} We consider the $5$-dimensional state space that has been introduced in Section \ref{mdp2}. The learning curve in Figure \ref{fig:learning} shows steady reward improvement with convergence around $50$k episodes, while the TD loss decreases smoothly and stabilizes at low values\footnote{Note that the reward initially decreases. Under our chosen parameter configuration (in particular, the episode length and the magnitude of the admissible actions), a nearly random policy in the early $\epsilon$-greedy phase trades aggressively enough to execute most of the inventory before the time horizon, so the initial reward is relatively high. As $\epsilon$ decreases and the policy becomes more structured, the agent temporarily learns to trade more cautiously, leaving a non-negligible inventory at maturity and thus suffering a larger terminal penalty, which explains the drop in reward. After this transient phase, the agent adjusts its strategy, and the reward increases and eventually converges.}. This indicates that the DDQN agent learned a stable and well-optimized policy throughout the training. 
\begin{figure}[H]
    \centering
    \begin{minipage}[t]{0.46\textwidth}
        \centering
        \includegraphics[width=\linewidth]{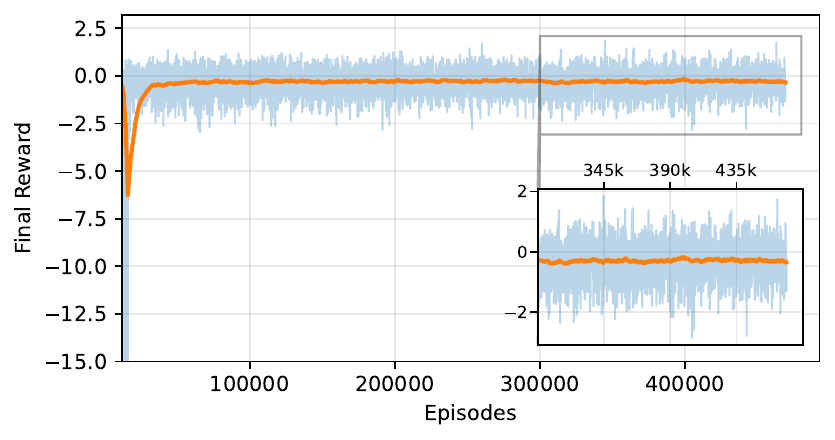}
        \subcaption{Final reward.}
        \label{fig:learning_curve}
    \end{minipage}\hspace{2em}
    \begin{minipage}[t]{0.46\textwidth}
        \centering
        \includegraphics[width=\linewidth]{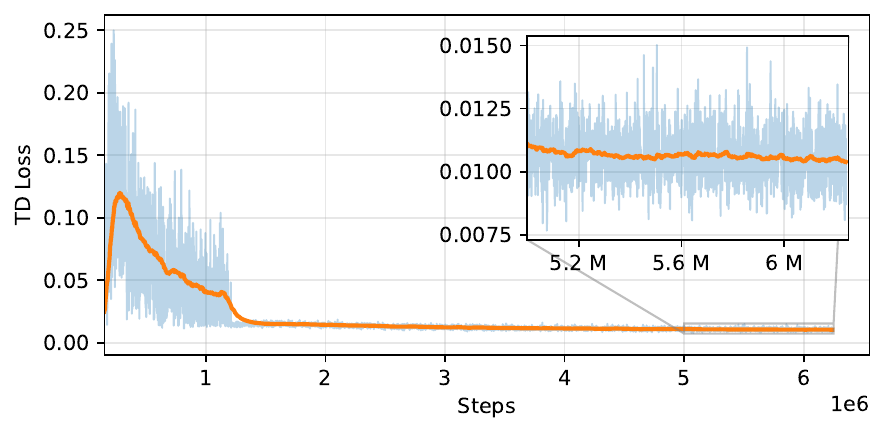}
        \subcaption{Temporal difference loss.}
        \label{fig:td_loss}
    \end{minipage}
    \caption{Learning across $6{,}240{,}000$ environment steps ($\approx 500{,}000$ episodes).}
    \label{fig:learning}
\end{figure}
To assess the quality of the learning, we analyze the Q-values for all actions in Figures \ref{fig:q_values_action0_c4}, \ref{fig:q_values_action1_c4}, and \ref{fig:q_values_action2_c4}. More precisely, we plot the Q-values of the trained RL agent as a function of inventory and time, when the ask price is equal to the arrival price and the bid and ask volumes are equal to their average values.
\begin{figure}[H]
    \centering
    \begin{minipage}[t]{0.38\textwidth}
        \centering
        \includegraphics[width=\linewidth]{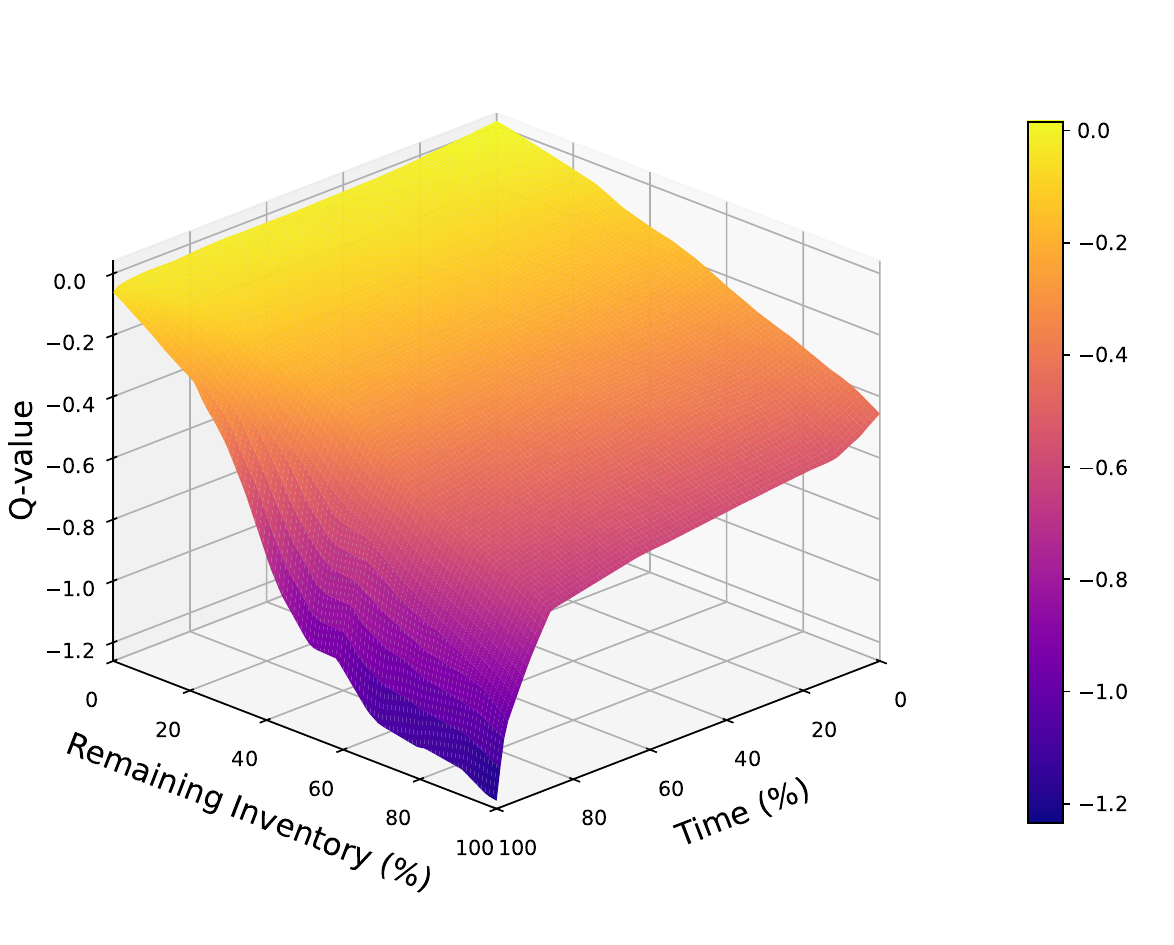}
        \subcaption{Q-values for action 0\%.}
        \label{fig:q_values_action0_c4}
    \end{minipage}\hspace{2.5em}
    \begin{minipage}[t]{0.38\textwidth}
        \centering
        \includegraphics[width=\linewidth]{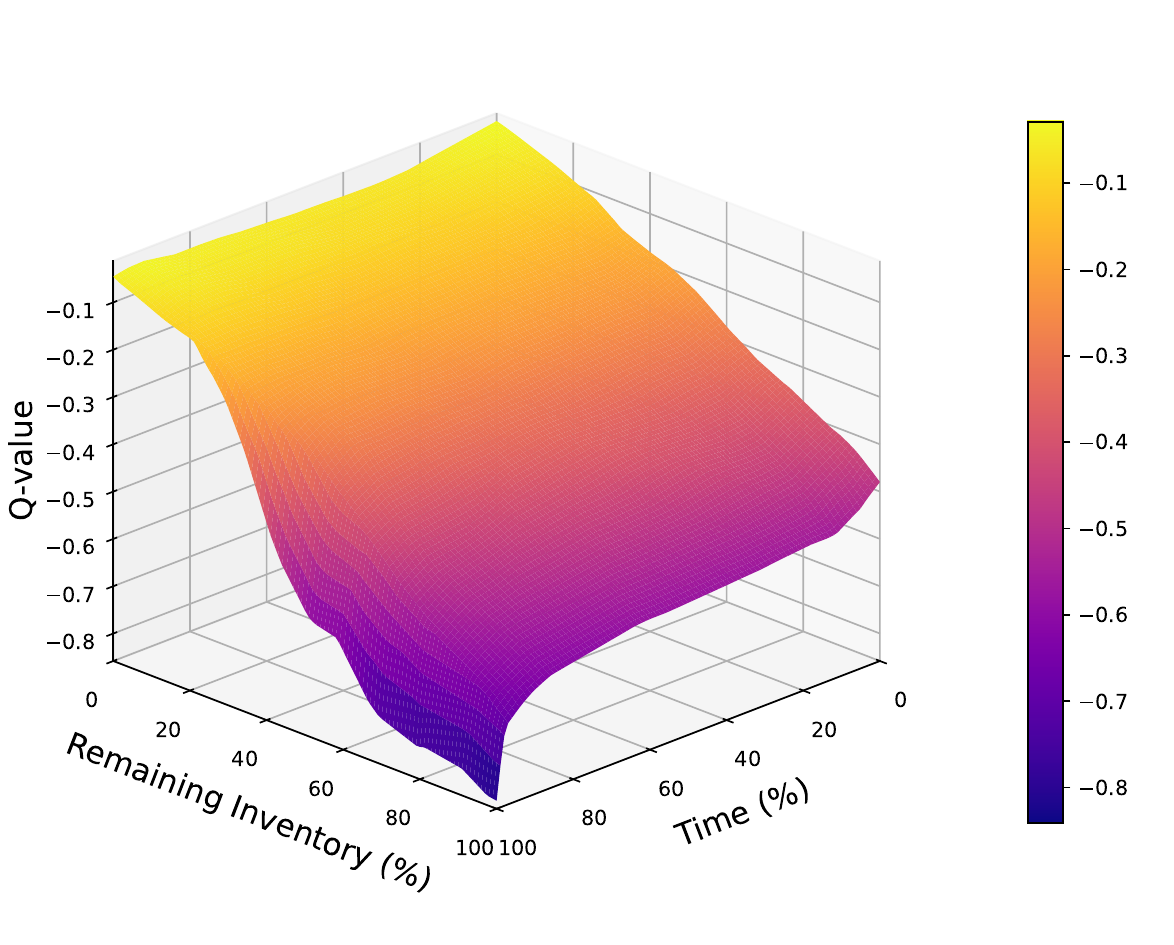}
        \subcaption{Q-values for action 50\%.}
        \label{fig:q_values_action1_c4}
    \end{minipage}
    \hspace{2.5em}
    \begin{minipage}[t]{0.38\textwidth}
        \centering
        \includegraphics[width=\linewidth]{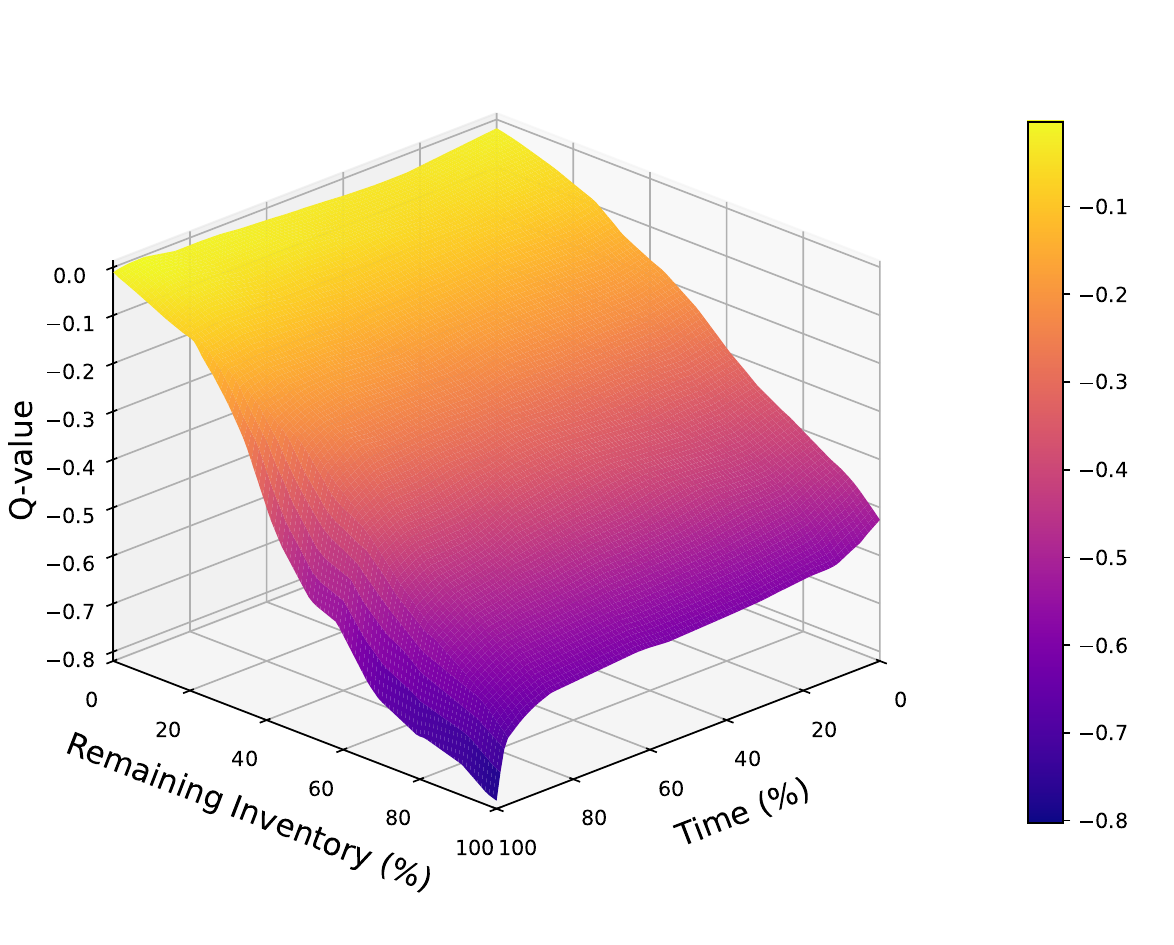}
        \subcaption{Q-values for action 100\%.}
        \label{fig:q_values_action2_c4}
    \end{minipage}
    \caption{Q-value surfaces at $P\!=\!P_0$ with mean bid and ask volumes.}
    \label{fig:q_values_actions_c4}
\end{figure}

All plots reveal that, for a fixed inventory, Q-values decline over time as the remaining horizon shortens, reflecting the agent’s increasing risk of incomplete execution, which is penalized in the reward function. The lowest Q-values are observed at the terminal time $T$, representing a lower bound on the execution cost. Furthermore, execution costs tend to rise with larger inventories, consistent with the greater market impact and urgency associated with executing larger positions. Moreover, execution costs tend to rise with larger inventories. 

\paragraph*{Performance Evaluation Against Baselines.}  
Table \ref{tab:results_c4} reports the performance of the DDQN agent compared to benchmark strategies. Since the reward function is defined as the negative of the IS, minimizing IS is equivalent to maximizing the reward.
\begin{table}[H]
\centering
\begin{tabular}{l|cccccc}
\hline
        & POPV1 & POPV2 & POPV3 & POPV4  & TWAP & DDQN  \\
\hline
Mean    & -0.343  & -0.342  & -0.400 & -0.399  & -0.365  & $\mathbf{-0.259}^{***}$  \\
Std     & \textbf{0.378}  & 0.472  & 0.388  & 0.437 & 0.652  & 0.631  \\
\hline
\end{tabular}
\caption{Reward results on $20{,}000$ test episodes.}
\label{tab:results_c4}
\end{table}
It is clear that the RL agent achieves the best overall performance, with a significantly higher average reward than all benchmarks. As shown in Figure \ref{fig:is_c4}, it matches TWAP’s best-case performance while maintaining limited worst-case losses.
\begin{figure}[H]
    \centering
    \includegraphics[width=0.4\textwidth]{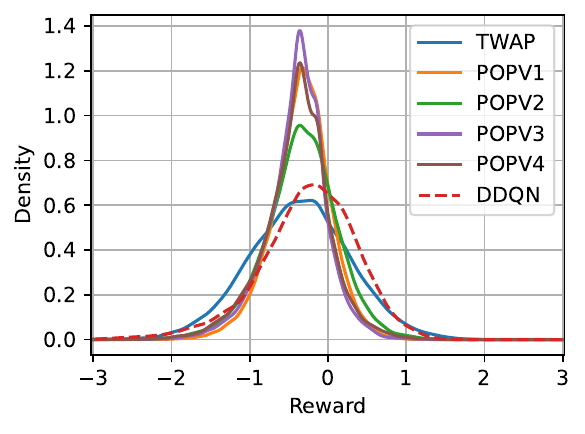}
    \caption{Reward distribution for the different tested strategies. }
    \label{fig:is_c4}
\end{figure}
We checked how often the different strategies are unable to complete the purchase within the time window. We find that the DDQN agent fails to fully complete the purchase in only $0.045\%$ of episodes, with an average of $1.44$ shares remaining. In comparison, POPV2 fails in $0.035\%$ of episodes, leaving an average of $2.14$ shares and POPV4 fails in 0.265\% of episodes, leaving an average of 2.51 shares. Thus, we can safely conclude that the considered strategies almost always complete the trading program.

\paragraph{Trading Patterns and Tactical Adaptation.}
Optimal execution problems are typically addressed using a two-layer framework consisting of strategy and tactic. The strategy component determines the overall trading schedule, namely the number of shares to execute within each time interval (see Figure \ref{fig:traj_inv_c4}). The tactic component, on the other hand, governs how these scheduled orders are executed within each interval (see Figures \ref{fig:gaps_variance_c4}
). The following figures report results averaged over 20,000 test episodes.
 \begin{figure}[H]
    \centering
        \centering
        \includegraphics[width=0.4\linewidth]{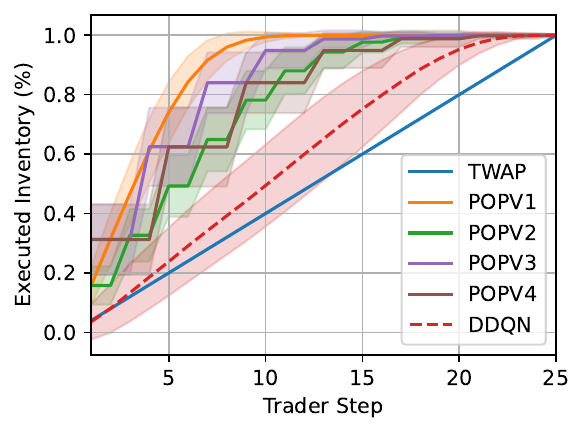}
        \captionof{figure}{Average execution trajectory of the different tested strategies as a function of time (measured in trader step).}
        \label{fig:traj_inv_c4}
    \end{figure}

Figure \ref{fig:traj_inv_c4} displays the average execution trajectory of the RL agent compared with those of TWAP and POPV. The DDQN trajectory remains approximately linear for most of the trading horizon, becoming slightly concave only near completion. This pattern contrasts with the more uniform or front-loaded behavior of the rule-based benchmarks and highlights the dynamic adjustment of execution pace observed in the DDQN runs. 

The figure, however, should be properly interpreted as it might convey the wrong impression that the DDQN strategy is static (as the TWAP or, more generally, the AC solution). First of all we show that the time needed to complete the execution with the DDQN strategy is highly variable. Figure \ref{fig:ep_length_c4} shows the distribution of the metaorder execution time, i.e. the number of trader steps required to complete the execution under the DDQN policy. Short episodes correspond to periods of abundant liquidity at the best ask, whereas longer episodes occur when the agent waits for more favorable trading opportunities. This broad distribution contrasts with the sharply peaked episode length profiles of benchmark strategies (see Figure \ref{fig:ep_length_benchmarks}), underscoring the adaptive nature of the learned policy. Finally, the decline in density a few steps before the time horizon indicates that the RL agent has correctly internalized the terminal penalty, completing execution by the time horizon.

\begin{figure}[H]
        \centering
        \includegraphics[width=0.4\linewidth]{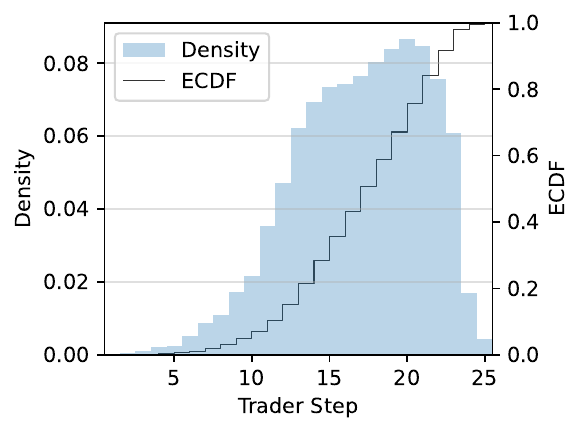}
        \captionof{figure}{Distribution of the metaorder execution time in number of trader steps.}
        \label{fig:ep_length_c4}
\end{figure}

Second, even restricting to episodes with the same metaorder execution length, we observe a large variability in the strategy. To show this, we compute for each metaorder execution time the average and variance of the gaps between consecutive executions, as shown in Figure \ref{fig:gaps_variance_c4} (gaps are measured in time steps: a gap of zero indicates two executions occurred at consecutive time steps).
\begin{figure}[H]
  \centering
  \begin{minipage}[t]{0.36\textwidth}
    \centering
    \includegraphics[width=\linewidth]{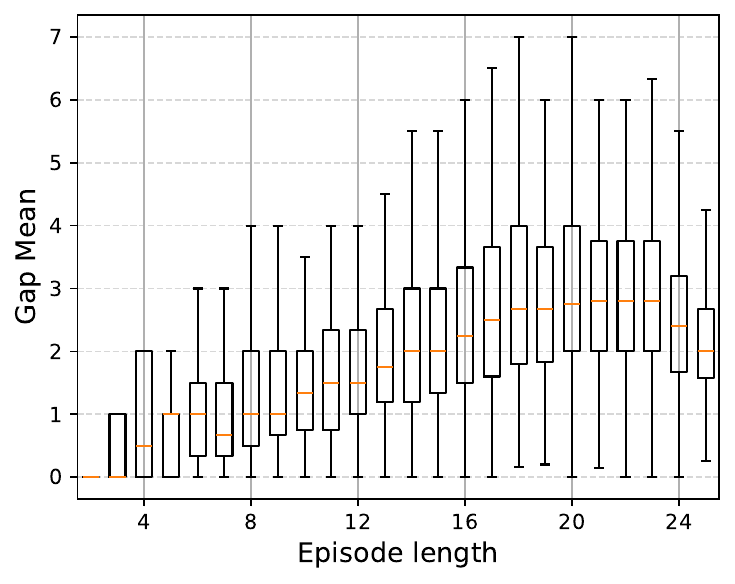}
  \end{minipage}
  \hspace{2em}
  \begin{minipage}[t]{0.36\textwidth}
    \centering
    \includegraphics[width=\linewidth]{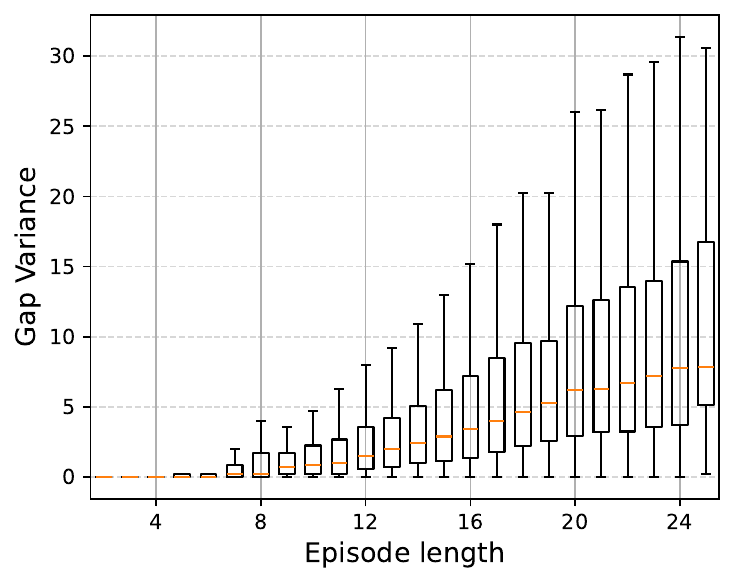}
  \end{minipage}

  \caption{Boxplots of average gaps (left) and gap variance (right) between consecutive executions per episode length for the DDQN strategy.}
  \label{fig:gaps_variance_c4}
\end{figure}
The broad interquartile ranges and long whiskers indicate that the RL agent does not trade at uniform intervals: even for episodes of comparable duration, the spacing between trades varies substantially. This variability becomes more pronounced for longer episodes, suggesting that the agent adjusts its trading frequency dynamically over time rather than following a fixed temporal schedule. Such heterogeneity across episodes of equal length implies that execution timing depends strongly on the prevailing market conditions. The resulting behavior is therefore tactical in nature, as the agent adapts its actions in response to short-term liquidity and price dynamics rather than adhering to a predetermined strategic rhythm.

To conclude, the DDQN trajectory exhibits a distinctly nonlinear, state-dependent profile: execution is accelerated under favorable conditions and decelerated as the horizon shortens. This pattern reflects a learned balance between immediacy and market impact, driven by the temporal structure of Q-value updates and declining continuation value near maturity. The resulting behavior demonstrates that the DDQN policy learns both strategic and tactical dimensions of execution, dynamically adapting to market states beyond the capability of fixed benchmark strategies (see \cite{reinforcement_nevmyvaka_2006}).

\paragraph{Feature Importance.}
To identify the main drivers of the RL agent’s decisions, we perform an input-gradient analysis  and the results are shown in Table \ref{tab:fi_actions_combined_c4} (see also Appendix \ref{app:feat_att}).
\begin{table}[H]
    \centering
    \begin{tabular}{l|c|c|c}
        \hline
        Feature & \makecell{Gradient \\ (Action $0\%$)} & \makecell{Gradient \\ (Action $50\%$)} & \makecell{Gradient \\ (Action $100\%$)} \\
        \hline
        Inventory   & 0.40 & 0.37 & 0.41 \\
        Ask Price   & 0.34 & 0.34 & 0.35 \\
        Time        & 0.10 & 0.07 & 0.06 \\
        Ask Volume  & 0.07 & 0.07 & 0.07 \\
        Bid Volume  & 0.06 & 0.05 & 0.04 \\
        \hline
    \end{tabular}
    \caption{Input-gradient for actions $0\%$, $50\%$ and $100\%$.}
    \label{tab:fi_actions_combined_c4}
\end{table}
The results indicate that the ask price and inventory are the primary determinants of action selection, suggesting that the agent mainly reacts to immediate market conditions. In contrast, the influence of time is small and decreases with action magnitude. The larger the executed volume, the less significant the time feature, reinforcing that the learned policy is tactical and adaptive rather than schedule-driven. These findings are consistent with a complementary SHAP value analysis\footnote{Results are available upon request.}.

\subsection{Robustness to Different Market Conditions}

To test robustness, the RL policy trained at $(\theta, \theta^{\mathrm{reinit}})\!=\!(0.7, 0.85)$ is evaluated across QRM simulations with $\theta, \theta^{\mathrm{reinit}} \in [0.5, 1.0]$. Figure \ref{fig:robustness} shows that it consistently exceeds the best benchmark, achieving up to 27\% higher performance. This demonstrates that the learned policy generalizes well to different market regimes and can be relied upon to maintain strong performance under varying conditions.
\begin{figure}[H]
    \centering
    \includegraphics[width=0.42\textwidth]{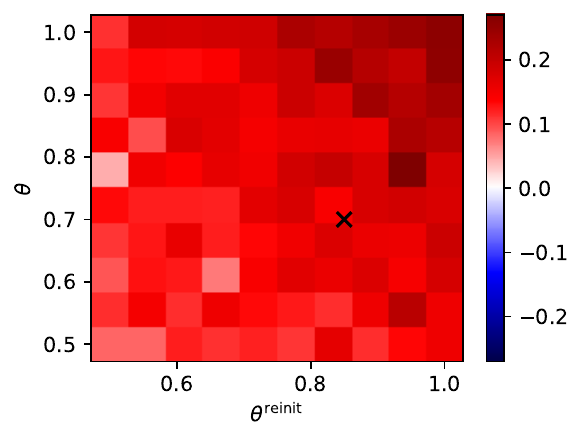}
    \caption{Heatmap of the relative difference between the average reward of the RL agent vs the best benchmark (TWAP in all the cases), averaged over 10,000 simulations. }
    \label{fig:robustness}
\end{figure}
\section{Conclusion}
In this work, we proposed a reinforcement-learning framework for optimal execution within the Queue-Reactive Model (QRM). By leveraging the ergodicity and microstructural features of the QRM, including its state-dependent order-flow intensities, its stochastic queue dynamics, and its endogenous liquidity replenishment mechanisms, we produced a simulation environment that captures both direct liquidity consumption and indirect order-flow responses. Training a Double Deep Q-Network (DDQN) in this setting shows that the RL agent learns execution strategies that are tactically adaptive and strategically robust. Across all configurations, the agent consistently outperforms standard benchmarks such as TWAP, adjusting to local liquidity and short-term price pressure rather than following a fixed schedule. Analyses of execution timing, Q-values, and feature importance confirm that the policy internalizes subtle microstructural patterns, including nonlinear impact and volume-imbalance effects.

Possible extensions include adding limit-order placement and queue-management decisions to balance passive and aggressive execution. Extending the methodology to multi-asset execution problems would introduce cross-impact effects and portfolio-level constraints, broadening the scope of the approach. Another potential direction is to integrate predictive alpha signals into the state space so that the agent jointly optimizes execution and directional positioning, thereby connecting optimal trading and optimal execution.

\section*{Acknowledgements} 
Yadh Hafsi acknowledges the support of the Chaire Risque Financiers, Société Générale, at École Polytechnique. 

\printbibliography

\newpage
\appendix

\section{Feature Importance}
\label{app:feat_att}

We draw $N$ transitions $(s,a)$ from the replay buffer and compute the corresponding Q-values $Q(s,a)$. For each input feature $s_i$, we evaluate the gradient
$g_i = \frac{\partial Q(s,a)}{\partial s_i},$
take its absolute value, and average across all $N$ sampled transitions. The resulting per-feature scores quantify how variations in each state variable influence the Q-value. This gradient-based method requires only a single backward pass per state–action pair.

\section{Additional Plots}
\label{app:additional_plots}

\subsection{Invariant Distribution}
\begin{figure}[H]
    \centering
    \includegraphics[width=0.34\textwidth]{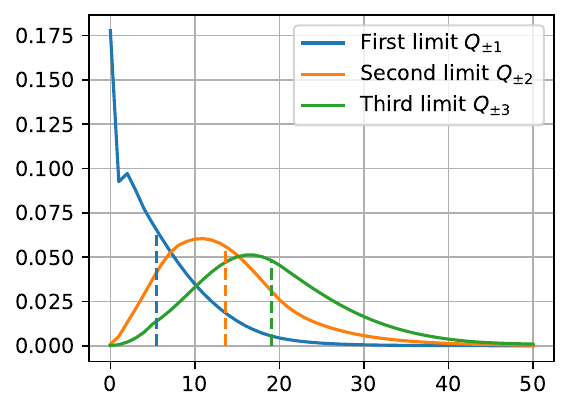}
    \caption{Invariant distribution of $Q_{\pm 1}, Q_{\pm 2}, Q_{\pm 3}$, taken from \cite{huang_2015}.}
    \label{fig:inv_distrib}
\end{figure}  

\subsection{Market Impact}\label{app:refills}
 We provide additional details on the bid and ask refill scenarios mentioned in Section \ref{sec:mi_refill} to clarify the mechanism behind the price mean reversion observed after sending a market order that consumes the entire best ask. In a bid-refill scenario (see Figure \ref{fig:combined_market_impact_c}), price initially increases because the most probable subsequent event after buying the best ask is a limit order placement at the best bid, which raises the mid-price from $100.010$ to $100.015$. Logically, the same phenomenon happens when $\theta\!=\!0$ and we observe that the trajectories coincide over short horizons. However, simulations indicate that this new liquidity typically vanishes rapidly, causing the mid-price to return to its previous level. When the reference price then decreases to $100.005$ (with probability $\theta$) and volumes are redrawn from the invariant distribution (with probability $\theta\theta^{\mathrm{reinit}}$), the average mid-price settles half a tick below its initial value. This mechanism occurs frequently enough to produce systematic mean reversion, driven by the interaction between the calibrated intensities and $(\theta, \theta^{\mathrm{reinit}})$. In contrast, when $\theta\!=\!0$, the mid-price increases by half a tick, as expected. 
\begin{figure}[H]
  \centering
  \begin{minipage}[t]{0.36\textwidth}
    \centering
    \includegraphics[width=\linewidth]{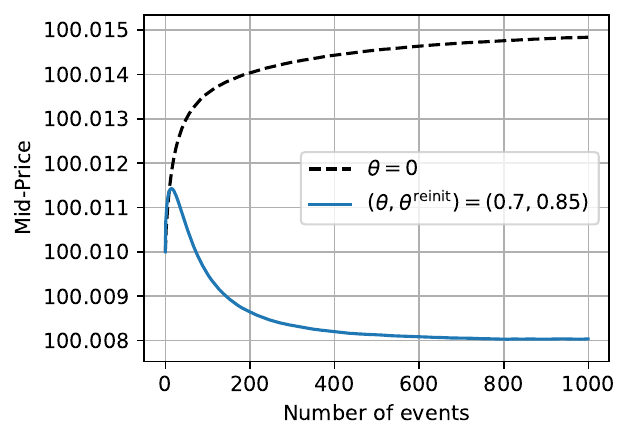}
    \subcaption{Bid refill following a buy order with $(\theta, \theta^{\mathrm{reinit}})\!=\!(0.7, 0.85)$.}
    \label{fig:combined_market_impact_c}
  \end{minipage}
  \hspace{2.2em}
  \begin{minipage}[t]{0.36\textwidth}
    \centering
    \includegraphics[width=\linewidth]{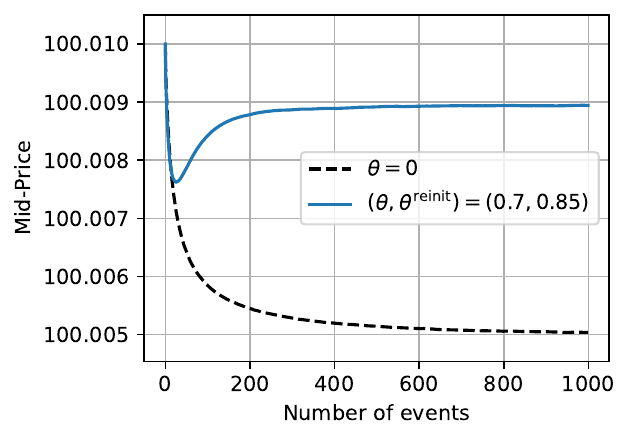}
    \subcaption{Ask refill following a buy order with $(\theta, \theta^{\mathrm{reinit}})\!=\!(0.7, 0.85)$.}
    \label{fig:combined_market_impact_d}
  \end{minipage}

  \caption{Evolution of $\mathbb{E}[p_k]$ after buying the best ask at $k\!=\!0$, averaged over $1{,}000{,}000$ simulations. The black dotted line denotes the reference case $\theta\!=\!0$. The mid-price before the trade is $100.005$.}
\end{figure}
 In the ask-refill scenario (see Figure \ref{fig:combined_market_impact_d}), the opposite sequence occurs. These two price response patterns are consistently observed across all tested parameter combinations. However, the aggregate results (see Figure \ref{fig:combined_market_impact}) vary according to the relative weighting of each of these bid-ask refill scenarios, determined by the probabilities $(\theta, \theta^{\mathrm{reinit}})$.

 \subsection{POPV Benchmarks}
\begin{figure}[H]
  \centering

  \begin{minipage}[t]{0.36\textwidth}
    \centering
    \includegraphics[width=\linewidth]{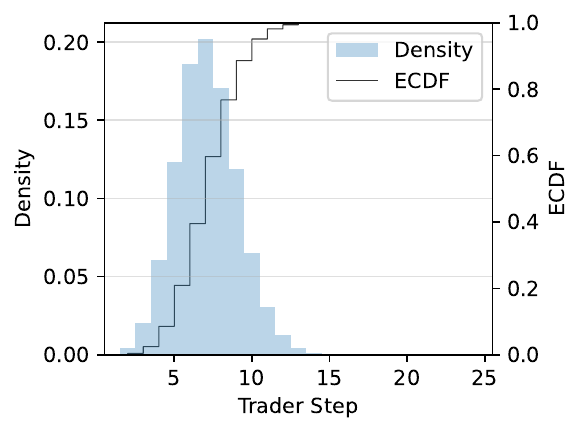}
    \subcaption{POPV1.}
    \label{fig:popv1}
  \end{minipage}
  \hspace{2em}
  \begin{minipage}[t]{0.36\textwidth}
    \centering
    \includegraphics[width=\linewidth]{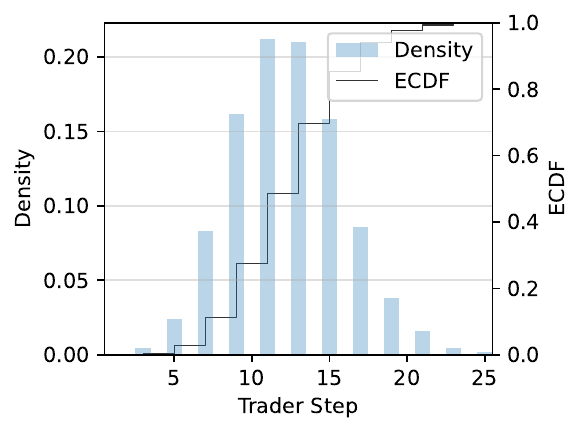}
    \subcaption{POPV2.}
    \label{fig:popv2}
  \end{minipage}

  \begin{minipage}[t]{0.36\textwidth}
    \centering
    \includegraphics[width=\linewidth]{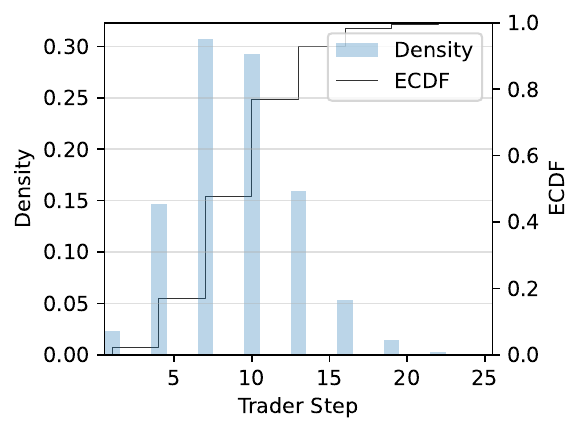}
    \subcaption{POPV3.}
    \label{fig:popv3}
  \end{minipage}
  \hspace{2em}
  \begin{minipage}[t]{0.36\textwidth}
    \centering
    \includegraphics[width=\linewidth]{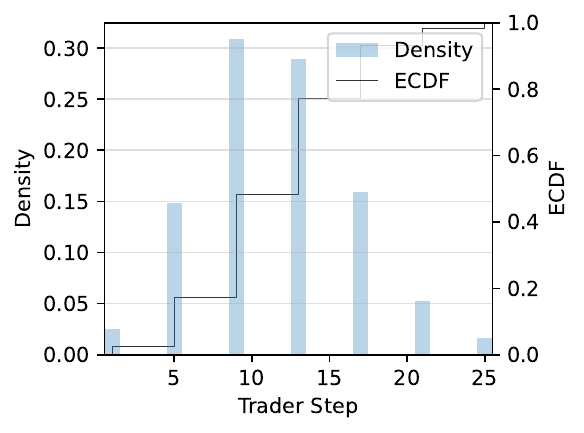}
    \subcaption{POPV4.}
    \label{fig:popv4}
  \end{minipage}

  \caption{Episode length distribution of the POPV benchmarks.}
  \label{fig:ep_length_benchmarks}
\end{figure}

\section{Additional Experiments with Reduced State Spaces} \label{app:reduced_state_spaces}

We evaluate multiple state-space configurations to examine how the dimensionality of the input representation influences the agent’s performance. These configurations help isolate the contribution of individual market features such as best ask volume and bid-ask imbalance. The following results correspond to the case of a binary action space, $\mathcal{A}\!=\!\{0\%, 100\%\}$. For feature importance, we report only the input-gradient analysis, as similar patterns were observed with the SHAP-value analysis. The following figures report results averaged over 20,000 test episodes.
\subsection{Reduced Model: 3D State and Binary Action Space}

In this first configuration, we consider a minimal 3-dimensional state space that includes the remaining inventory, the time and the best ask price. We report the results of the trained RL agent and of the benchmarks in Table \ref{tab:results_c1} and show the reward distribution in Fig. \ref{fig:is_c1}. The RL agent has the second best average reward after TWAP. We observe that TWAP exhibits a much larger variance in rewards compared to DDQN. This stems from the fact that TWAP is entirely agnostic to price and order book dynamics: it performs poorly when prices trend upward and favorably when they decline, resulting in high variability in realized rewards. Thus, the performance of the RL agent is not satisfactory: we would expect it to perform better than TWAP as it is better informed.

\begin{table}[H]
\centering
\begin{tabular}{l|cccccc}
\hline
        & POPV1   & POPV2 & POPV3 & POPV4  & TWAP & DDQN  \\
\hline
Mean    & -0.413  & -0.408  & -0.400  & -0.399 & $\mathbf{-0.365}^{**}$ & -0.386  \\
Std     & \textbf{0.279}  & 0.342  & 0.388  & 0.437 & 0.652 & 0.473 \\
\hline
\end{tabular}
\caption{Reward results. POPV1, POPV2, POPV3, TWAP and DDQN fully execute on all episodes. POPV4 has 2.51 shares remaining in average on 53 episodes (0.265\%).}
\label{tab:results_c1}
\end{table}

\begin{table}[H]
    \centering
    \begin{tabular}{l|c|c}
        \hline
        Feature & \makecell{Gradient \\ (Action $0\%$)} & \makecell{Gradient \\ (Action $100\%$)} \\
        \hline
        Inventory  & 0.45 & 0.42 \\
        Ask Price  & 0.25 & 0.26 \\
        Time       & 0.19 & 0.10 \\
        \hline
    \end{tabular}
    \caption{Input-gradient for actions $0\%$ and $100\%$.}
    \label{tab:fi_actions_combined_c1}
\end{table} 

\begin{figure}[H]
    \centering
    \begin{minipage}{0.3\textwidth}
        \centering
        \includegraphics[width=\linewidth]{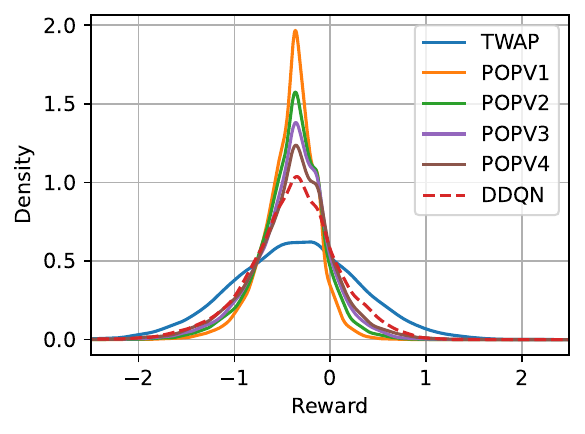}
        \caption{Reward distribution.}
        \label{fig:is_c1}
    \end{minipage}
    \hfill
    \begin{minipage}{0.3\textwidth}
        \centering
        \includegraphics[width=\linewidth]{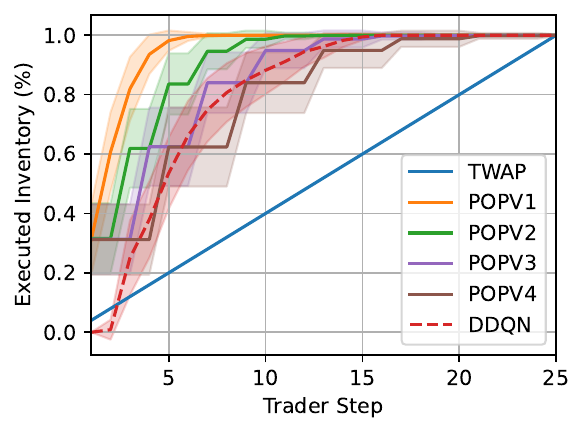}
        \caption{Average inventory trajectory.}
        \label{fig:traj_inv_c1}
    \end{minipage}
    \hfill
    \begin{minipage}{0.3\textwidth}
        \centering
        \includegraphics[width=\linewidth]{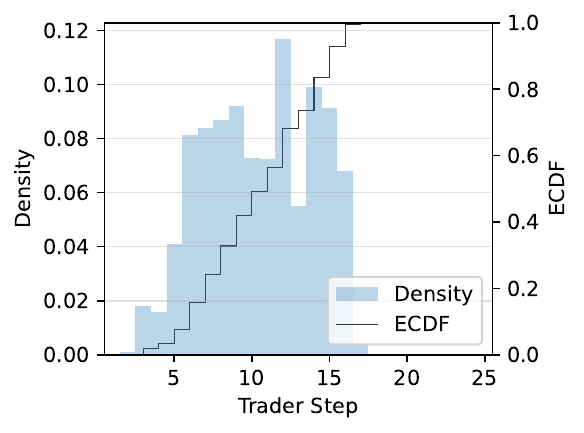}
        \caption{Episode length distribution.}
        \label{fig:ep_length_c1}
    \end{minipage}
\end{figure}

\subsection{Reduced Model: 4D State and Binary Action Space}

In this second configuration, we consider a 4-dimensional state space that includes the remaining inventory, the time, the best ask price and best ask volume. The rationale for including the best ask volume is that it provides the agent with information about the number of shares that would be executed if it decided to consume the entire best ask volume. This feature is particularly relevant because the distribution of best ask volumes is highly right-skewed and exhibits a heavy tail. Thus, the agent can take advantage of unusually large volumes to execute substantial trades in a single action. The results are reported in Table \ref{tab:results_c2}: the RL agent outperforms all the benchmarks in terms of average IS. Figure \ref{fig:traj_inv_c2} shows that the episodes are longer compared to those of the 3-dimensional state-space agent, indicating that the agent has learned to act more patiently in order to exploit more favorable future prices.

\begin{table}[H]
\centering
\begin{tabular}{l|cccccc}
\hline
        & POPV1   & POPV2 & POPV3 & POPV4  & TWAP & DDQN  \\
\hline
Mean    & -0.413  & -0.408  & -0.400  & -0.399 & -0.365 & $\mathbf{-0.325}^{***}$  \\
Std     & \textbf{0.279}  & 0.342  & 0.388  & 0.437 & 0.652 & 0.594  \\
\hline
\end{tabular}
\caption{Reward results. POPV1, POPV2, POPV3 and TWAP fully execute on all episodes. POPV4 has 2.51 shares remaining in average on 53 episodes (0.265\%). DDQN has 6.77 shares remaining in average on 31 episodes (0.155\%).}
\label{tab:results_c2}
\end{table}

\begin{table}[H]
    \centering
    \begin{tabular}{l|c|c}
        \hline
        Feature & \makecell{Gradient \\ (Action $0\%$)} & \makecell{Gradient \\ (Action $100\%$)} \\
        \hline
        Inventory  & 0.33 & 0.31 \\
        Ask Price  & 0.25 & 0.26 \\
        Time       & 0.12 & 0.07 \\
        Ask Volume & 0.06 & 0.06 \\
        \hline
    \end{tabular}
    \caption{Input-gradient for actions $0\%$ and $100\%$.}
    \label{tab:fi_actions_combined_c2}
\end{table}

\begin{figure}[H]
    \centering
    \begin{minipage}[t]{0.3\textwidth}
        \centering
        \includegraphics[width=\linewidth]{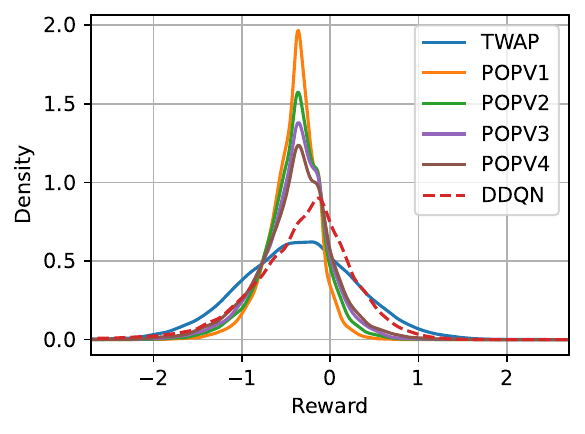}
        \caption{Reward distribution.}
        \label{fig:is_c2}
    \end{minipage}
    \hfill
    \begin{minipage}[t]{0.3\textwidth}
        \centering
        \includegraphics[width=\linewidth]{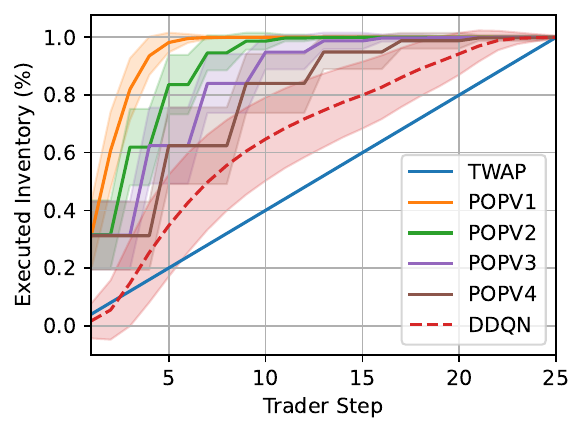}
        \caption{Average inventory trajectory.}
        \label{fig:traj_inv_c2}
    \end{minipage}
    \hfill
    \begin{minipage}[t]{0.3\textwidth}
        \centering
        \includegraphics[width=\linewidth]{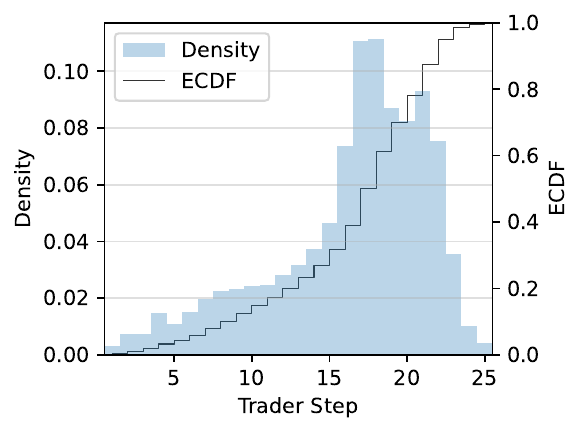}
        \caption{Episode length distribution.}
        \label{fig:ep_length_c2}
    \end{minipage}
\end{figure}

\subsection{Reduced Model: 5D State and Binary Action Space}

In this third configuration, we consider a 5-dimensional state space that includes the remaining inventory, the time, the best ask price and the best bid-ask volume. The rationale behind adding the best bid volume is that it informs the agent of the volume imbalance, a well-established short-term predictor of price movements in market microstructure. The results are reported in Table \ref{tab:results_c3}: the RL agent outperforms all the benchmarks.

\begin{table}[H]
\centering
\begin{tabular}{l|cccccc}
\hline
        & POPV1 & POPV2 & POPV3 & POPV4  & TWAP & DDQN  \\
\hline
Mean    & -0.413  & -0.408  & -0.400 & -0.399  & -0.365  & $\mathbf{-0.290}^{***}$  \\
Std     & \textbf{0.279}  & 0.342  & 0.388  & 0.437 & 0.652  & 0.541  \\
\hline
\end{tabular}
\caption{Reward results. POPV1, POPV2, POPV3 and TWAP fully execute on all episodes. POPV4 has 2.51 shares remaining in average on 53 episodes (0.265\%). DDQN has 2.33 shares remaining in average on 3 episodes (0.015\%).}
\label{tab:results_c3}
\end{table}

\begin{table}[H]
\centering
\begin{minipage}{\textwidth}
    \centering
    \begin{tabular}{l|c|c}
        \hline
        Feature & \makecell{Gradient \\ (Action $0\%$)} & \makecell{Gradient \\ (Action $100\%$)} \\
        \hline
        Inventory  & 0.44 & 0.40\\
        Ask Price  & 0.29 & 0.30 \\
        Time       & 0.12 & 0.06 \\
        Ask Volume & 0.07 & 0.07  \\
        Bid Volume & 0.06 & 0.03  \\
        \hline
    \end{tabular}
    \caption{Input-gradient for actions $0\%$ and $100\%$.}
    \label{tab:fi_action_0}
\end{minipage}
\end{table}

\begin{figure}[H]
    \centering
    \begin{minipage}[t]{0.3\textwidth}
        \centering
        \includegraphics[width=\linewidth]{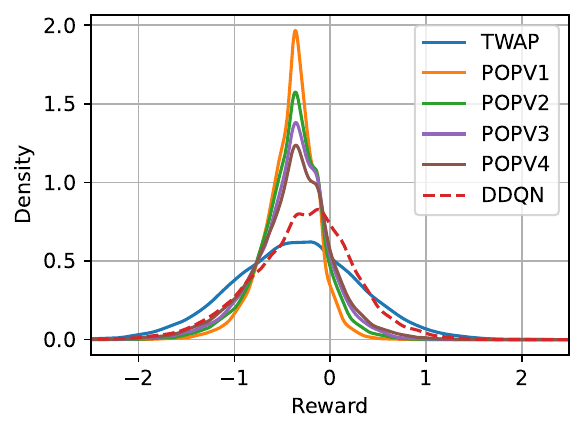}
        \caption{Reward distribution.}
        \label{fig:is_c3}
    \end{minipage}
    \hfill
    \begin{minipage}[t]{0.3\textwidth}
        \centering
        \includegraphics[width=\linewidth]{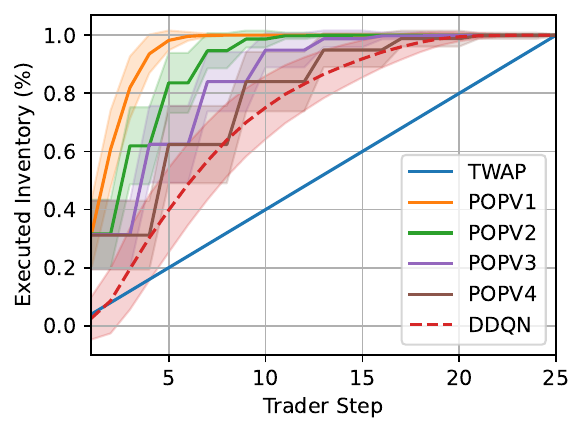}
        \caption{Average inventory trajectory.}
        \label{fig:traj_inv_c3}
    \end{minipage}
    \hfill
    \begin{minipage}[t]{0.3\textwidth}
        \centering
        \includegraphics[width=\linewidth]{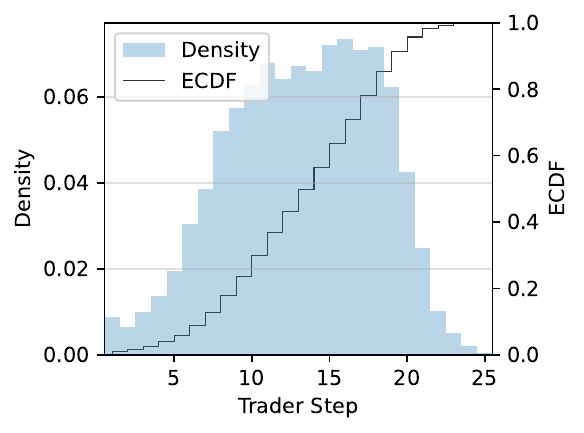}
        \caption{Episode length distribution.}
        \label{fig:ep_length_c3}
    \end{minipage}
\end{figure}

\begin{figure}[H]
  \centering
  \begin{minipage}[t]{0.36\textwidth}
    \centering
    \includegraphics[width=\linewidth]{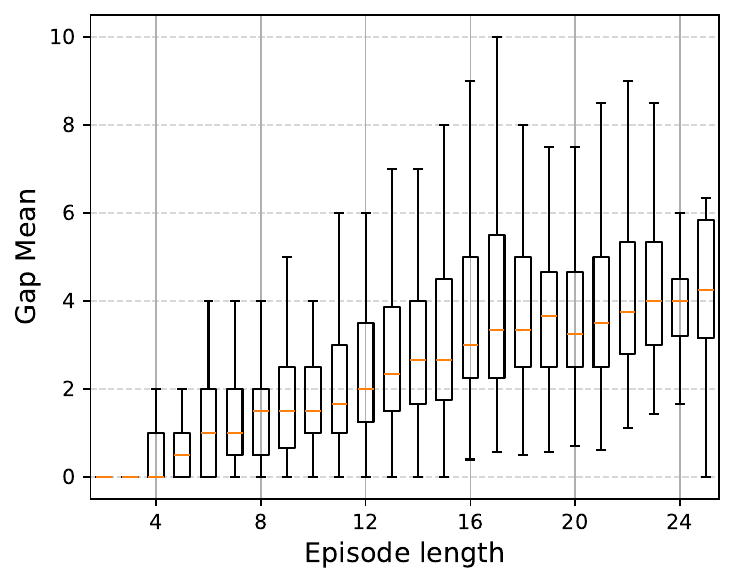}
  \end{minipage}
  \hspace{2em}
  \begin{minipage}[t]{0.36\textwidth}
    \centering
    \includegraphics[width=\linewidth]{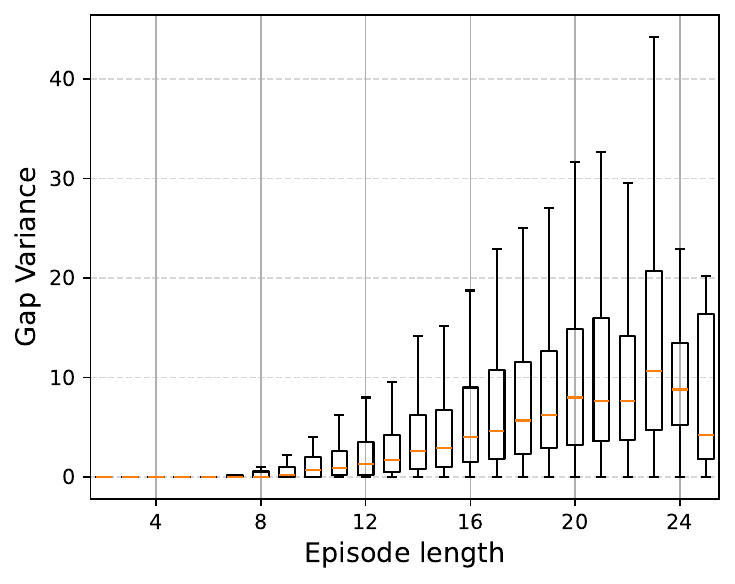}
  \end{minipage}\

  \caption{Boxplots of average gaps (left) and gap variance (right) between consecutive executions per episode length.}
  \label{fig:gaps_combined}
\end{figure}

\begin{figure}[H]
    \centering
    \begin{minipage}[t]{0.4\textwidth}
        \centering
        \includegraphics[width=\linewidth]{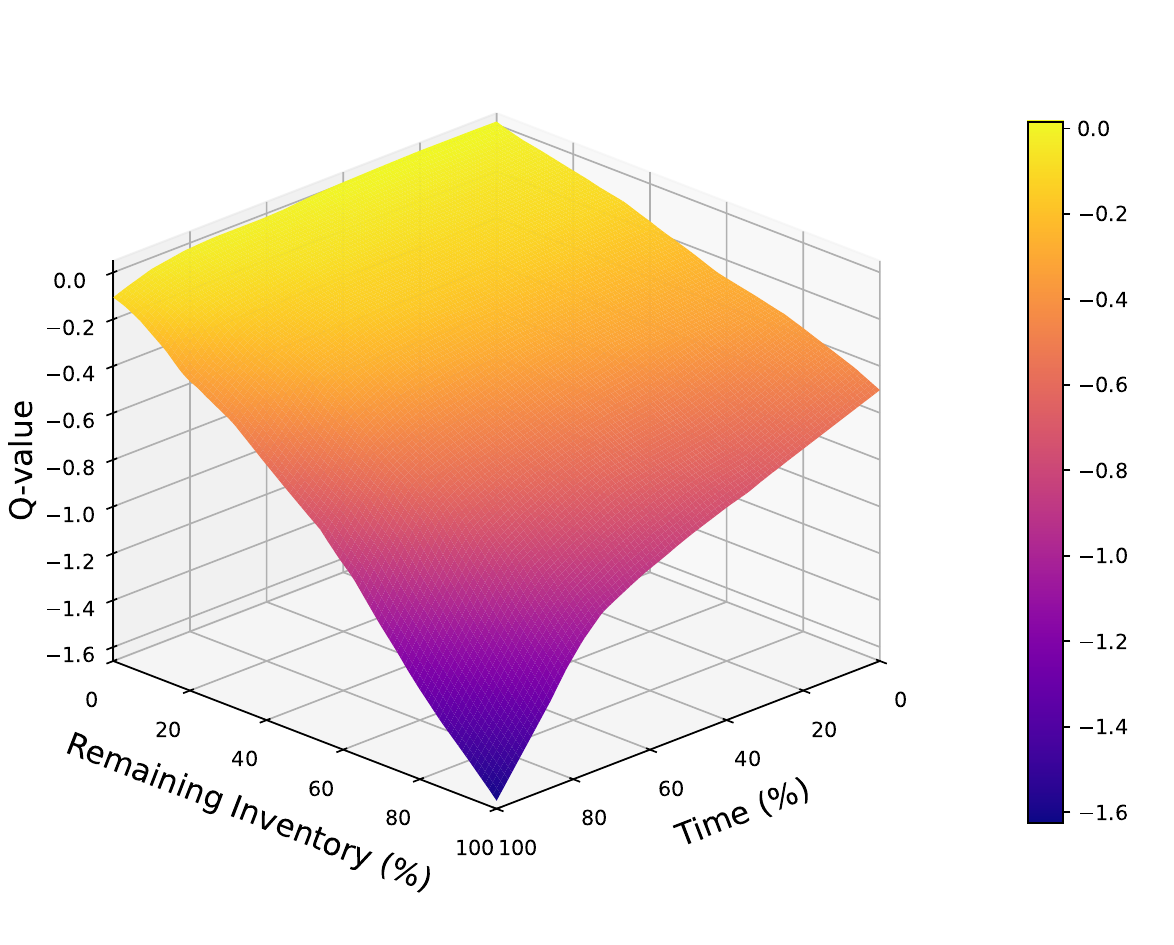}
        \subcaption{Q-values for action 0\%.}
        \label{fig:q_values_action0_c3}
    \end{minipage}\hspace{2.5em}
    \begin{minipage}[t]{0.4\textwidth}
        \centering
        \includegraphics[width=\linewidth]{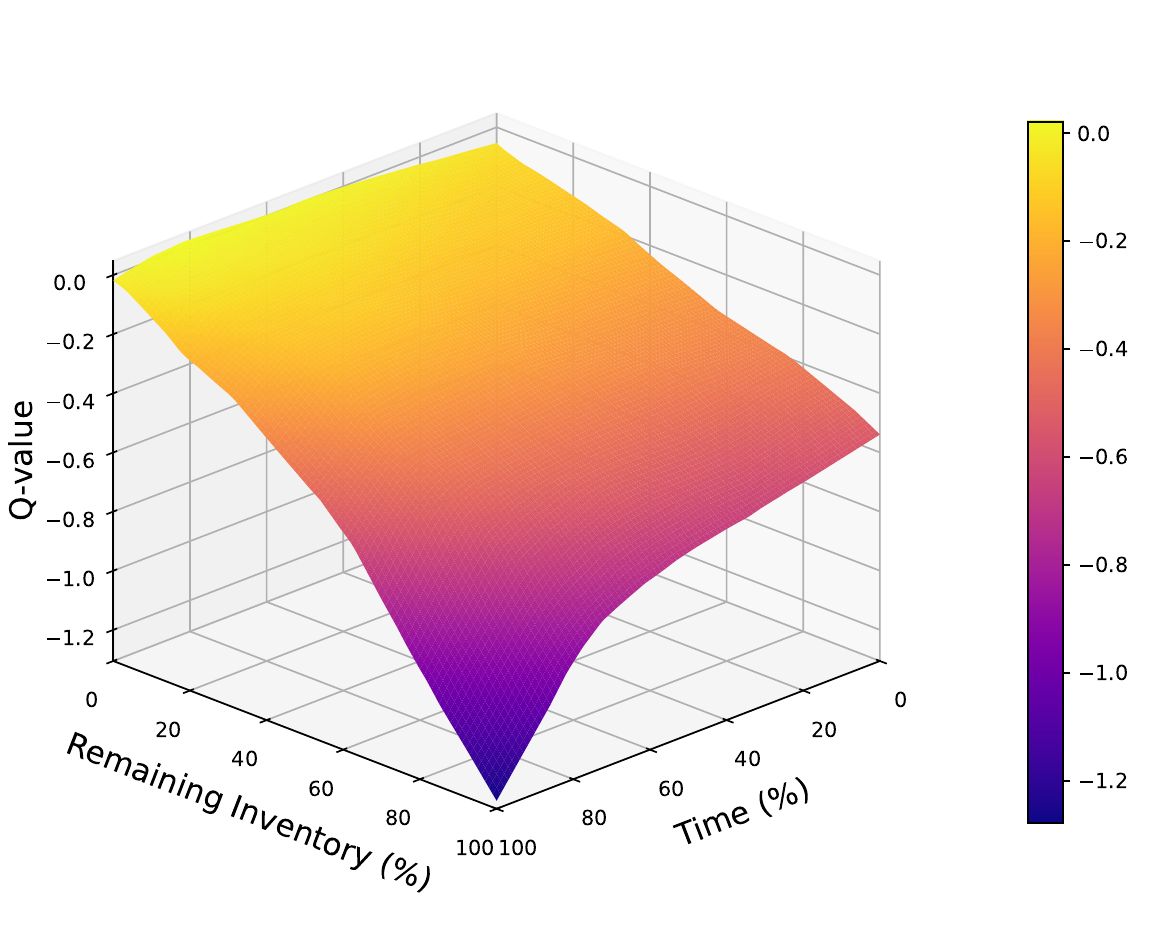}
        \subcaption{Q-values for action 100\%.}
        \label{fig:q_values_action1_c3}
    \end{minipage}
    \caption{Q-value surfaces when the price equals $P_0$ and the best bid/ask volumes are equal to their means.}
    \label{fig:q_values_actions_c3}
\end{figure}

\subsection{Comparing Models with Different State and Action Spaces}\label{sec:compa_analysis}

We summarize and compare the RL performances across all combinations of state and action spaces explored. For the 2-dimensional action space $\mathcal{A} = \{0\%, 100\%\}$, the algorithms DDQN1, DDQN2, and DDQN3 correspond to state spaces of dimensions 3, 4, and 5, respectively. The model denoted DDQN4 refers to the best-performing agent trained with a 3-dimensional action space $\mathcal{A} = \{0\%, 50\%, 100\%\}$ and a 5-dimensional state representation. Our study shows that the RL agent's performance gradually increases by extending the state and action space.

\begin{table}[H]
\centering
\begin{tabular}{l|ccccc}
\hline
         & TWAP & DDQN1   & DDQN2  & DDQN3 & DDQN4 \\
\hline
Mean    & -0.365 & -0.386  & -0.325  & -0.290 & $\mathbf{-0.259}^{***}$ \\
Std     & 0.652 & \textbf{0.473}  & 0.594  & 0.541 & 0.631 \\
\hline
\end{tabular}
\caption{Reward results for different state and action space dimensions.}
\label{tab:results_summary}
\end{table}

\begin{figure}[H]
    \centering
    \begin{minipage}[t]{0.38\textwidth}
        \centering
        \includegraphics[width=\linewidth]{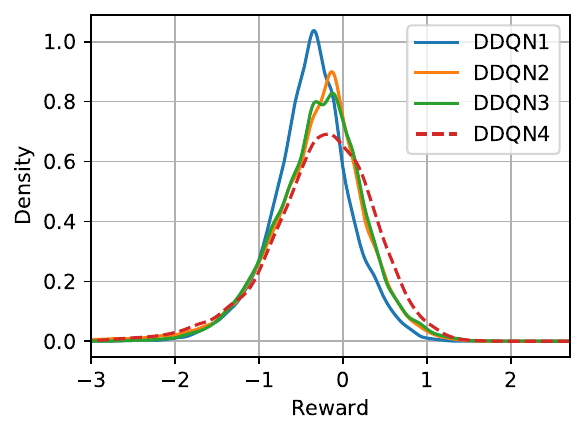}
        \subcaption{Histogram.}
        \label{fig:is_summary_hist}
    \end{minipage}\hspace{2.5em}%
    \begin{minipage}[t]{0.42\textwidth}
        \centering
        \raisebox{1.5em}{\includegraphics[width=\linewidth]{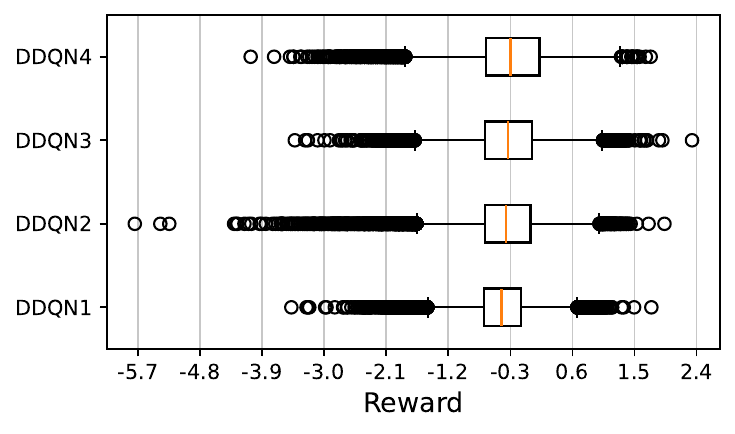}} 
        \subcaption{Boxplot showing the mean.}
        \label{fig:is_summary_boxplot}
    \end{minipage}
    \caption{Reward distribution.}
    \label{fig:is_summary}
\end{figure}

\end{document}